\documentclass[11pt]{article}
\usepackage[utf8]{inputenc}
\usepackage{amsmath,amsthm,amssymb,fullpage}
\usepackage[colorlinks=true,linkcolor=blue,citecolor=blue,linktocpage=true]{hyperref}
\usepackage{algorithm,algpseudocode}
\usepackage{microtype}
\usepackage{graphicx}

\newcommand{\ex}[2]{{\ifx&#1& \mathbb{E} \else \underset{#1}{\mathbb{E}} \fi \left[#2\right]}}
\newcommand{\pr}[2]{{\ifx&#1& \mathbb{P} \else \underset{#1}{\mathbb{P}} \fi \left[#2\right]}}
\newcommand{\exc}[3]{{\ifx&#1& \mathbb{E} \else \underset{#1}{\mathbb{E}} \fi \left[ #2 \middle| #3 \right]}}
\newcommand{\prc}[3]{{\ifx&#1& \mathbb{P} \else \underset{#1}{\mathbb{P}} \fi \left[ #2 \middle| #3 \right]}}
\newcommand{\var}[2]{{\ifx&#1& \mathbb{V} \else \underset{#1}{\mathbb{V}} \fi \left[#2\right]}}

\DeclareMathOperator*{\argmin}{arg\,min}

\DeclareRobustCommand{\[}{\begin{equation}} % make displayed equations using \[\] be numbered
\DeclareRobustCommand{\]}{\end{equation}}

\newtheorem{theorem}{Theorem}
\newtheorem{lemma}[theorem]{Lemma}
\newtheorem{definition}[theorem]{Definition}
\newtheorem{corollary}[theorem]{Corollary}
\newtheorem{proposition}[theorem]{Proposition}

\newtheorem{claim}[theorem]{Claim}
\newtheorem{openproblem}[theorem]{Open Problem}

\numberwithin{theorem}{section}
\numberwithin{equation}{section}
\numberwithin{algorithm}{section}
\numberwithin{table}{section}

\usepackage{multicol}
\makeatletter
\renewcommand\tableofcontents{%
    \section*{\makebox[\linewidth][c]{\contentsname}%
      \@mkboth{\MakeUppercase\contentsname}{\MakeUppercase\contentsname}}%
    \begin{multicols}{2}%
    \@starttoc{toc}%
    \end{multicols}
    }
\makeatother

\usepackage[style=alphabetic,backend=bibtex,maxalphanames=10,maxbibnames=20,maxcitenames=10,giveninits=true,doi=false,url=true,backref=true]{biblatex}

\newcommand*{\citet}[1]{\AtNextCite{\AtEachCitekey{\defcounter{maxnames}{999}}}\textcite{#1}}
\newcommand*{\citep}[1]{\cite{#1}}

\title{Privately Estimating Black-Box Statistics}
\author{
G\"unter F. Steinke\thanks{University of Canterbury ~\dotfill~ \texttt{gunter.steinke@canterbury.ac.nz}}
\and
Thomas Steinke\thanks{Google DeepMind ~\dotfill~ \texttt{steinke@google.com}}
}

\date{}

\begin{document}
\maketitle
\footnotetext{Accepted for presentation at \href{https://learningtheory.org/colt2026/}{COLT 2026}. Full version: \url{https://arxiv.org/abs/2510.00322}.}
%\vspace*{-10pt}
\begin{abstract}
    Standard techniques for differentially private estimation, such as Laplace or Gaussian noise addition, require guaranteed bounds on the sensitivity of the estimator in question. But such sensitivity bounds are often large or simply unknown. Thus we seek differentially private methods that can be applied to arbitrary black-box functions. A handful of such techniques exist, but all are either inefficient in their use of data or require evaluating the function on exponentially many inputs. In this work we present a scheme that trades off between statistical efficiency (i.e., how much data is needed) and oracle efficiency (i.e., the number of evaluations). We also present lower bounds showing the near-optimality of our scheme.
\end{abstract}
\vspace{15pt}
\begin{small}
\tableofcontents
\end{small}
\newpage

\section{Introduction}
Differential privacy \citep{DMNS06} provides a mathematical framework for measuring and controlling the leakage of sensitive information via computations on a private dataset.
%Formally, an algorithm $M$ satisfies $(\varepsilon,\delta)$-differential privacy \citep{DKMMN06} if, for all pairs of inputs $x,x'$ differing only by the addition or removal of one person's data and all measurable sets of outputs $S$, we have
%\begin{equation}
%    \pr{}{M(x) \in S} \le e^\varepsilon \pr{}{M(x') \in S} + \delta .
%\end{equation}
Given a function $f$ that we wish to evaluate on a private dataset $x$, the simplest and best-known method for ensuring differential privacy is to add Laplace or Gaussian noise -- e.g., $M(x)=f(x)+\xi$, where $\xi \gets \mathsf{Laplace}(\Delta_f/\varepsilon)$ is random noise scaled according to the privacy parameter $\varepsilon$ and the (global) sensitivity of $f$. The \emph{global sensitivity of $f$} is given by $\Delta_f := \sup_{x,x'} |f(x)-f(x')|$, where the supremum is over all pairs of inputs differing only by the addition or removal of one person's data.

In many cases, the global sensitivity of the function we want to evaluate is large (or even infinite), or simply unknown. %(because the function is given to us as a ``black box'' or is too complicated to analyse). 
In practice, the function may be given to us as a ``black box'' -- that is, we can only evaluate the function as an ``oracle'' and not inspect its inner workings, or it may be presented as a piece of untrusted code that is too complicated to analyse.\footnote{As a motivating practical example of code that is too complicated to analyze, consider $f$ such that computing $y=f(x)$ requires training a machine learning model with the input $x$ being the training data and the output $y$ being the predicted label of a fixed unlabelled data point \cite{pate1,pate2}.} %Thus black-box methods would be highly desirable in practice.
In these cases, we cannot rely on the standard noise addition approach to ensure differential privacy.
Many methods have been proposed for privately evaluating functions that go beyond noise addition and are applicable to functions with high (global) sensitivity; see Section \ref{sec:related} for a brief survey. 
However, all of these methods have various drawbacks (which we briefly discuss next) that have limited practical adoption: Either they require evaluating the function many times (or they require some a priori structural knowledge about the function), or they are statistically inefficient.

Most methods for evaluating functions with high global sensitivity still rely on being able to compute or bound some property of the function, such as smooth sensitivity \citep{NRS07} or distance to instability \citep{DL09}, instead of global sensitivity. %We are interested in methods that are applicable in a black-box manner. 
Hence these methods require either careful analysis of the function or evaluating the function over a large fraction of its domain (i.e., they may require evaluating the function on exponentially many -- or even infinitely many -- inputs). 
This makes these methods impractical in the black-box setting. %, where we can evaluate the function, but otherwise have no structural knowledge of the function. %In practice, the function may be given as a piece of untrusted code, rather than a black box, but often the code is too complicated to analyse. Thus black-box methods would be highly desirable in practice.

A further limitation of the aforementioned methods is that they require evaluating the function on arbitrary inputs that may not correspond to any realistic data. This may ``break'' the function in the sense that the function may be well-behaved on real data, but could produce arbitrary values if even one input datapoint is corrupted. For example, changing one input can change the mean of a dataset arbitrarily. That is to say, the \emph{local} sensitivity of the function may be large, which implies that, e.g., the smooth sensitivity will also be large.
%This makes these methods impractical in truly black-box settings.
This limitation can be circumvented by using \emph{down-local} algorithms (as we do); down-local algorithms only evaluate the function on (subsets of) the given input 
\cite{cummings2020individual,FangDY22,kohli2023differential,LRSS25}.
However, most down-local algorithms still require evaluating a black-box function on exponentially many subsets of the input.

There is one method that does not have the aforementioned limitations: The sample-and-aggregate framework of \citet{NRS07} only requires evaluating the function on a small number of inputs (with each input consisting only of ``real'' data) and does not require any structural assumptions about the function, which makes it appealing in practice \cite{pate1,pate2,cohen2023hot}. The catch is that this method is statistically inefficient in terms of its use of data.
That is, if we start with a dataset of size $n$, then sample-and-aggregate evaluates the function on datasets of size $O(\varepsilon n)$. (Specifically, the sample-and-aggregate framework partitions the dataset into $O(1/\varepsilon)$ equal-sized parts.)
Very roughly, the final accuracy of our private estimate given a dataset of size $n$ is only as good as a non-private estimate on a dataset of size $O(\varepsilon n)$.
Thus, if the privacy parameter $\varepsilon$ is small, then sample-and-aggregate suffers a significant cost in terms of statistical accuracy.

\subsection{Our Contributions}
In this article, we examine the tradeoff between statistical efficiency (i.e., how much data is needed to estimate a statistic) and oracle query complexity (i.e., how many times we need to evaluate the function).

Our main result is a differentially private algorithm which takes a real-valued black-box function $f$ and a private dataset $x$ and evaluates the function on multiple subsets of the input dataset and then outputs an estimate $y$ for the value of the function. 
Our algorithm interpolates between sample-and-aggregate \citep{NRS07} and the more recent computationally inefficient algorithm of \citet{LRSS25}.

\textbf{Statistical View on Accuracy:}
Our work differs from much of the prior work in how we quantify the accuracy of our estimate. Namely, most prior work attempts to ensure $y \approx f(x)$. However, this goal can be too narrow. %\footnote{Note that without the privacy constraint guaranteeing statistical accuracy can non-trivial. We can always output $y=f(x)$ non-privately, but we must worry about overfitting/generalization when estimating properties of the distribution from which $x$ was drawn.}
In many settings, the input $x$ consists of independent and identically distributed (i.i.d.)~samples from some distribution $\mathcal{D}$ and our goal is to estimate properties of the distribution $\mathcal{D}$ rather than of the sample $x$.
Thus we take a statistical view of accuracy. Namely, we assume that the input $x$ consists of i.i.d.~samples and the function $f$ returns a good estimate with high probability when it is given enough i.i.d.~samples; our algorithm then also produces a good estimate.

\begin{theorem}[Main Result] \label{thm:main}
    Let $\mathcal{Y} \subseteq \mathbb{R}$ be finite and let $\mathcal{X}$ be arbitrary; denote $\mathcal{X}^* = \bigcup_{n \in \mathbb{N}} \mathcal{X}^n$.
    Let $\varepsilon,\delta>0$ and $n, m, t \in \mathbb{N}$ satisfy \[n \ge m \ge t=\frac{1}{\varepsilon} \log(1/\delta) \exp(O(\log^* |\mathcal{Y}|)). \label{eq:thm:main:t}\] Let\footnote{We precisely define $C(n,m,t)$ later in Definition \ref{def:covering_design}; for now the given upper bound suffices. Throughout, $\log$ denotes the natural logarithm with base $e \approx 2.718$.} \[k = C(n,m,t) < \frac{{n \choose t}}{{m \choose t}} \left(1 + \log{ m \choose t} \right) +1 \le O\left( \left(\frac{ne}{m}\right)^t \cdot t \log m\right). \label{eq:thm:main:k}\]
    Then, for all $f : \mathcal{X}^* \to \mathcal{Y}$, there exists an algorithm $M^f : \mathcal{X}^n \to \mathcal{Y}$ with the following properties.
    \begin{itemize}
        \item \textbf{Privacy:} $M^f$ is $(\varepsilon,\delta)$-differentially private. 
        \item \textbf{Statistical Accuracy:} Let $\mathcal{D}$ be an arbitrary probability distribution on $\mathcal{X}$. Suppose $\pr{X \gets \mathcal{D}^{n-m}}{|f(X)-\nu|\le\alpha}\ge 1-\beta$ for some $\alpha, \beta, \nu \in \mathbb{R}$. Then $\pr{X \gets \mathcal{D}^n}{|M^f(X)-\nu|\le\alpha}\ge 1- k\beta$.
        \item \textbf{Oracle Efficiency:} On input $x \in \mathcal{X}^n$, $M^f(x)$ selects $k$ subsets of $x$, each of size $n-m$, and evaluates $f$ on those subsets; other than these $k$ evaluations, $M^f(x)$ does not depend on either $f$ or $x$.
    \end{itemize}
\end{theorem}

Before continuing we make some remarks interpreting Theorem \ref{thm:main}:
\begin{enumerate}
    \item Informally, the statistical accuracy guarantee says that, with $n$ private samples, we can get the same accuracy as we could get with $n-m$ non-private samples. (There is an additional factor $k$ blowup in the failure probability, but this is secondary.) Intuitively, the parameter $m$ is the number of samples ``wasted'' to ensure differential privacy.

    \item Making $m$ smaller translates to better statistical accuracy, but it increases $k$. Making $m$ larger makes $k$ smaller, which means the algorithm requires fewer evaluations of the function $f$ -- i.e., lower oracle complexity. This tradeoff is the key phenomenon we study. 
    In particular, the following are three points on the tradeoff curve:
    \begin{enumerate}
        \item Setting $m=\frac{t}{t+1}n$ yields $n-m=\frac{n}{t+1}$ and $k \le t+1$,\footnote{This value of $k=C(n,\frac{t}{t+1}n,t) \le t+1$ is tighter than the upper bound given in the theorem statement; see Equation \ref{eq:samp-agg-design}. For this informal discussion, we ignore integer divisibility issues.} which corresponds to sample-and-aggregate \citep{NRS07}. That is, the cost of privacy is a \emph{multiplicative} factor of $t$ in the sample complexity -- i.e., with $n$ private samples we get accuracy comparable to $n-m=\Theta(n/t)$ non-private samples. This is the most computationally efficient instantiation of our result.
        \item Setting $m=t$ yields $k={n \choose t}$, which corresponds to the results of \citet{LRSS25}. In this setting, the cost of privacy is an \emph{additive} $t$ samples, at the expense of the number of evaluations $k$ being exponential in $t$. This is the most statistically efficient instantiation of our result.
        \item The above points (a and b) are the extremes on the tradeoff curve; now we consider a setting that interpolates between these: 
        Setting $m = \frac{tn}{t+c}$ yields $k \le O(\min \{ e^{t+c} \cdot t\log m, t^c\})$. (Here $c\ge 1$ is an appropriate integer.) Relative to (a) sample-and-aggregate, this increases the number of data points $n-m = \frac{cn}{t+c}$ available for each evaluation by a factor of almost $c$ -- i.e., we go from $n-m=\Theta(n/t)$ to $n-m=\Theta(c n / t)$ -- at the expense of a modest increase in the number of evaluations $k$. This parameter setting is likely of practical interest.
    \end{enumerate}

    \item Note that the accuracy parameters $\alpha$ and $\beta$ are not inputs to the algorithm. In a sense, our algorithm ``automatically adjusts'' to the difficulty of the problem.
    
    \item The value $t$ in Equation \ref{eq:thm:main:t} depends on the privacy parameters $\varepsilon,\delta$ and on the size of the output space $\mathcal{Y}$. The dependence on the privacy parameters $t \ge \frac{\log(1/\delta)}{\varepsilon}$ is essentially the best we could hope for (see the lower bound below). And the dependence on the size of the output space is extremely mild; $\log^*$ denotes the iterated logarithm, which grows extremely slowly.

    \item Our algorithm has an oracle efficiency guarantee, but not an overall computational efficiency guarantee. That is, we bound the computational cost related to evaluating the function, but not the cost of choosing the subsets and processing the values. \label{nb:computational-efficiency} See Section \ref{sec:limitations} for further discussion of these limitations.
\end{enumerate}

We show that the number of evaluations of the oracle function $f$ (i.e., $k$ in Theorem \ref{thm:main}) is roughly optimal:

\begin{theorem}[Lower Bound] \label{thm:lower}
    Let $M^f : \mathbb{Z}^n \to \mathbb{Z}$ be a randomised algorithm that makes at most $k$ queries to an oracle $f : \mathbb{Z}^* \to \mathbb{Z}$.
    Suppose that, for every oracle $f$, the algorithm $M^f$ satisfies $(\varepsilon,\delta)$-differential privacy and the following.
    Let $\mathcal{D}$ be an arbitrary distribution on $\mathbb{Z}$ and let $\nu \in \mathbb{Z}$. If $\pr{X \gets \mathcal{D}^{n-m}}{f(X)=\nu} \ge 0.999$, then $\pr{X \gets \mathcal{D}^n}{|M^f(X)-\nu| \le 1} \ge 1/2$.
    Then we must have \[k \ge \frac{{n \choose t}}{{m \choose t}} \cdot \Omega(1) ~~~~ \text{ with } ~~ t = \Theta(1/\varepsilon) \label{eq:thm:lower:1}\] 
    and, simultaneously, \[k \ge \frac{{n \choose t}}{{m \choose t}} \cdot \Omega(\delta^{0.01}) ~~~~ \text{ with } ~~ t = \Theta(\log(1/\delta)/\varepsilon), \label{eq:thm:lower:2}\] assuming $\delta \le (\varepsilon/10)^{1.1} \le 1$.
\end{theorem}

Contrasting the lower bounds in Equations \ref{eq:thm:lower:1} and \ref{eq:thm:lower:2} with the upper bound in Equation \ref{eq:thm:main:k}, we see that the combinatorial term $k \approx \frac{{n \choose t}}{{m \choose t}}$ is present in both the upper and lower bounds. There are some additional factors, but these are relatively minor. The main difference is that the upper bound uses $t=\frac{1}{\varepsilon} \log(1/\delta) \exp(O(\log^* |\mathcal{Y}|))$, while the lower bound sets $t = \Omega(1/\varepsilon)$ or $t = \Omega(\log(1/\delta)/\varepsilon)$.
Thus there is a multiplicative gap in the parameter $t$ depending on the size of the output space. This is potentially significant, since this multiplicative factor affects the number of queries $k$ in an exponential fashion.
We remark that some dependence on the size of the range $\mathcal{Y}$ is known to be necessary for statistical estimation \citep{bun2015differentially,alon2019private}. Thus, while the $\exp(O(\log^* |\mathcal{Y}|))$ term in the upper bound could potentially be improved, it cannot be removed entirely.

The lower bound is a strong result. It shows that the number of evaluations must be impractically large in many parameter regimes.
The message is that working in the black-box setting inherently incurs a high cost either in terms of statistical efficiency or oracle efficiency.

\subsection{Our Techniques}

\textbf{Algorithm:}
Our algorithm can be viewed as an extension of the sample-and-aggregate paradigm \cite{NRS07}. Namely, we evaluate the function on subsets of the input and then we aggregate those values in a way that ensures differential privacy. The novelty is in how we choose the subsets and how we aggregate the values.

Our algorithm has two main technical ingredients: 
First we use a combinatorial object known as a covering design or a Tur\'an system to select $k$ overlapping subsets of the input on which to evaluate the function.
Second, following \citet{LRSS25}, we use a variant of the shifted inverse mechanism of \citet{FangDY22} to aggregate the values in a differentially private manner.

The property of the covering design is that if $t$ out of the $n$ input datapoints are corrupted, then at least one of the $k$ subsets on which we evaluate the function will not contain any corrupted datapoints. (Note that this property holds without knowing which datapoints are corrupted.)
Intuitively, $(\varepsilon,\delta)$-differential privacy requires robustness to $t=O(\log(1/\delta)/\varepsilon)$ corruptions. And the covering design ensures this level of robustness, but only in the weak sense that, if $t$ inputs are corrupted, then at least one out of $k$ values is uncorrupted.\footnote{In contrast, for sample-and-aggregate we have the stronger robustness guarantee that, if $t$ inputs are corrupted, then at least $k-t$ out of $k$ values are uncorrupted.}

It only remains to aggregate the values in a way that translates this weak form of robustness into differential privacy.
Computing a differentially private mean or median of the $k$ values does not suffice, since a single datapoint could affect a majority of the values. This is where the shifted inverse mechanism fits in.

To illustrate how the shifted inverse mechanism works, consider the special case with a binary output space $\mathcal{Y}=\{0,1\}$. (The general case can, with some loss, be reduced to this case.)
Now we ask ``how many input datapoints would I need to remove so that the remaining output values are all 0?'' (That is, if we remove a datapoint, then all of the output values that depend on that datapoint are removed.) 
If all of the output values are $0$, then the answer to the query is $0$ -- i.e., no datapoints need to be removed.
If all of the output values are $1$, then the answer to the query is at least $t+1$, by the properties of the covering design. 
By construction, this query has sensitivity $1$; thus the answer to the query can be approximated in a differentially private manner by adding Laplace or Gaussian noise.
As long as $t$ is large enough, we can accurately distinguish between the case where all of the output values are $0$ and the case where all of them are $1$. (When some values are $0$ and some are $1$, the outcome is indeterminate.)
Assuming each value is individually accurate with high probability, the aggregated value will also be accurate with high probability.

\textbf{Example Instantiation:}
To give further intuition for how our algorithm works, we consider instantiating it with a particular simple covering design.
Again consider a binary output space $\mathcal{Y}=\{0,1\}$ and let $t = O\big(\frac1\varepsilon \log(1/\delta) \big)$.
Partition the dataset (of size $n$) into $t+2$ parts of (almost) equal size; call these $D_1, \cdots, D_{t+2}$.
Now we evaluate the function $f : \mathcal{X}^* \to \mathcal{Y}$ on all $\frac{(t+1)(t+2)}{2}$ pairs $D_i \cup D_j$. That is, we evaluate the function on subdatasets of size $2 \frac{n}{t+2}$; so each pair gets a label of $0$ or $1$.
(Note that partitioning into $t+2$ equal-sized parts and considering the union of pairs of parts corresponds to a $(n, m,t)$-covering design (Definition \ref{def:covering_design}) with $n-m = 2 \frac{n}{t+2}$.\footnote{To be precise, these pairs form a Tur\'an system, and their complements form a covering design.})

If each computation of $f(D_i \cup D_j)$ yields the correct label -- say, $0$ -- with probability $\ge 1-\beta$, then a union bound over the pairs tells us that, with probability $\ge 1- \frac{(t+1)(t+2)}{2} \beta$, all of the labels are $0$. Similarly, if instead the correct answer is $1$, then with high probability all the labels are $1$.

Now we apply the shifted inverse mechanism, which entails computing the minimum number of parts $D_1, \cdots, D_{t+2}$ that need to be chosen to cover all pairs that have a label of $1$. If all of the pairs have a label of $0$, then this number is zero; if all of the pairs have a label of $1$, then this number is $ \ge t+1$. Since this number is low sensitivity, it can be privately approximated, which yields the final output.

Note that computing this number corresponds to computing the size of a minimum vertex cover of the following graph: Each part $D_i$ corresponds to a vertex and, for each pair $i,j$, there is an edge between $D_i$ and $D_j$ if and only $f(D_i \cup D_j) = 1$. Since vertex cover is NP-complete \cite{Karp1972}, computing this value is potentially intractable. See Section \ref{sec:limitations} for further discussion of this limitation.

Comparing this example instantiation with sample-and-aggregate, we see that we evaluate the function $f$ on datasets of size approximately $2n/t$ rather than $n/t$, which yields better statistical accuracy. In exchange, we need approximately $t^2/2$ evaluations of the function $f$, rather than $t$. This is one instantiation of the tradeoff that our algorithm offers.

\textbf{Lower Bound:} 
Our lower bound mirrors the intuition for our algorithm. 
That is, $(\varepsilon,\delta)$-differential privacy requires robustness to $t=O(\log(1/\delta)/\varepsilon)$ corruptions. To be precise, we perform a packing argument \cite{hardt2010geometry}. That is, we use group privacy to argue that changing $t$ values cannot totally change the output of the algorithm.
We restrict the algorithm to evaluating the function on subsets of the input of size $n-m$. (We can force this restriction using standard tricks, such as making the function fail when given an input that isn't of this form.) 
Now the only way for the algorithm to satisfy the given level of robustness is for it to query enough sets so that at least one of them contains none of the $t$ corrupted points.
Roughly, this implies that the subsets queried by the algorithm must form a covering design. The lower bound then follows.

\section{Related Work}\label{sec:related}

There is a long line of work on differentially privately evaluating a function on a sensitive dataset that makes weaker and weaker assumptions on the underlying function, which we briefly survey in this section.
\citet{jha2013testing} were the first to explicitly study the black-box setting, where we make no assumptions about the function beyond what we can glean from evaluating it. 

\subsection{Alternatives to Global Sensitivity}
Adding Laplace or Gaussian noise scaled to global sensitivity has been the standard approach to ensure differential privacy since its inception \cite{DMNS06,DKMMN06}. And this approach is surprisingly versatile. However, there has been a long line of work seeking methods that circumvent the limitations of global sensitivity, to which the present article adds.
We now briefly survey the most closely related approaches.

\citet{NRS07} introduced two methods that go beyond global sensitivity -- smooth sensitivity and sample-and-aggregate (which we discuss later in this section).
Whereas global sensitivity depends only on the function and not on the dataset, \textbf{smooth sensitivity} is a dataset-dependent value. We can achieve differential privacy by adding noise that scales with the smooth sensitivity and the smooth sensitivity may be lower than the global sensitivity. 
In slightly more detail: Smooth sensitivity seeks to approximate the \emph{local sensitivity} -- i.e., how much can the function change by adding or removing one person's data to or from the actual dataset at hand.
For example, the median of a set of real numbers has unbounded global sensitivity, but its local sensitivity is bounded by the gap between the median value and the two values immediately before or after the median value in sorted order; indeed the local sensitivity of the median could even be zero if the median value is repeated multiple times.
Ideally, we could add noise scaled to the local sensitivity, rather than to the global sensitivity, but the scale of the noise could itself be sensitive.
Smooth sensitivity circumvents this issue by upper bounding the local sensitivity in a way that is itself low sensitivity (in a multiplicative, rather than additive, sense).
A major disadvantage of the smooth sensitivity approach is that it is often challenging to compute the smooth sensitivity, as it is still a global property of the function. In general, to compute the smooth sensitivity, we need to know the value of the function over its entire domain.

\citet{DL09} introduced the \textbf{propose-test-release} framework. Like smooth sensitivity, this framework seeks to exploit low local sensitivity. This framework begins a priori with a proposed upper bound on the local sensitivity. Then it performs a test to check whether or not this bound is correct. If the test passes, then it releases the value with noise scaled according to the bound. As long as the test is unlikely to yield a false positive, this guarantees differential privacy. A general recipe for the test step is to measure the distance (in terms of adding or removing people) from the given dataset to the nearest dataset with local sensitivity higher than the proposed bound. This distance is inherently low-sensitivity and so can be estimated privately. If the distance is large enough, the test will pass.
\citet{DL09} made the connection between differential privacy and robust statistics, which later work expanded upon \cite{kamath2019privately,bun2019average,avella2019differentially,kamath2020private,brunel2020propose,liu2021robust,brown2021covariance,ghazi2021robust,tsfadia2022friendlycore,hopkins2022efficient,liu2022differential,ashtiani2022private,sarathy2022analyzing,georgiev2022privacy,kothari2022private,asi2023robustness,alabi2023privately,liu2023near,canonne2023full,kamath2024broader,brown2025tukey}. 

The \textbf{inverse sensitivity mechanism} \cite{asi2020near,mir2011pan,johnson2013privacy,DPorg-inverse-sensitivity,asi2023robustness,hopkins2023robustness} provides a loss function with low global sensitivity for any function of interest; the exponential mechanism \cite{mcsherry2007mechanism} can then be applied to estimate the value of interest. The idea is simple: Given a function $f$, a dataset $x$, and a value $y$, define the loss $\ell(x,y) = \min\{ |x \setminus x'| + |x' \setminus x| : f(x')=y \}$ to be the least number of elements of $x$ that need to be added or removed in order to change the function value to $y$. Clearly, $\ell(x,y)=0$ if and only if $y=f(x)$. Furthermore, if the local sensitivity of $f$ at $x$ is small, then $\ell(x,y)\le 1$ implies $y \approx f(x)$. More generally, if $f$ is appropriately well-behaved near $x$, then any approximate minimiser $y$ of the loss $\ell(x,y)$ must be a good approximation to $f(x)$.

Another approach is to replace the function of interest with a function that has low global sensitivity and which still provides a good approximation to the original function.
For example, when computing a sum, we might clip the values to ensure that they are bounded; if the clipping threshold is chosen appropriately, this ensures low global sensitivity and doesn't change the value of the function too much.
In general, \textbf{Lipschitz extensions} provide low-sensitivity approximating functions that can be used in the differentially private setting \cite{kasiviswanathan2013analyzing,blocki2013differentially,raskhodnikova2015efficient,raskhodnikova2016lipschitz,raskhodnikova2016differentially}.
In particular, a \textbf{local Lipschitz filter} \cite{jha2013testing,awasthi2015limitations,lange2025local} provides a black-box method for constructing a Lipschitz extension. 

\paragraph{Limitations:}
All of the aforementioned methods (except sample-and-aggregate) suffer from computational intractability. They are only practical in special cases where we can analytically compute the relevant bounds. The underlying reason for this is that the quantities of interest all involve some universal quantification over datasets. In particular, they are not practical for functions that are given to us as a black box, as they would require enumerating exponentially many inputs -- or even infinitely many.

An additional limitation is that all of the above methods (except sample-and-aggregate) fail when the function $f$ has high \emph{local} sensitivity at the given dataset $x$. That is, there exists some input $x'$ neighbouring the given input $x$ such that the function value on that input $f(x')$ is far from the actual function value $f(x)$. In particular, adding one corrupted data point to the given input $x$ might ``break'' the function -- i.e. change the function value arbitrarily. For example, even the mean has infinite local sensitivity, when we allow unbounded inputs, even though it may be well-behaved for realistic inputs. (Note that high local sensitivity is an issue even if we are not in the black-box setting.)

\subsection{Down-Local Algorithms}
Recent work \cite{chen2013recursive,raskhodnikova2016differentially,cummings2020individual,FangDY22,kohli2023differential,LRSS25} has sought to overcome the aforementioned limitations by devising algorithms that are \textbf{down-local} in the sense that, given a function $f$ and a dataset $x$, the algorithm only evaluates $f$ on subsets of the given dataset $x$, rather than on arbitrary points in its domain.

This has two benefits: 
First, it restricts the number of inputs on which we can evaluate a black-box function $f$.
Second, down-local algorithms only evaluate the function on ``real'' data, which avoids the problem that some functions might ``break'' when given arbitrary inputs. In other words, down-local algorithms work well with functions that have high local sensitivity, but are well-behaved on realistic inputs.

\citet{cummings2020individual} effectively construct a Lipschitz extension by evaluating the function on all subsets of the input. This gives runtime which is ``only'' exponential in the number of input data points, and does not depend on the size of the function's domain. (In special cases, like the mean and median, they give polynomial-time algorithms.)

\citet{FangDY22} presented the \textbf{shifted inverse mechanism}, which is a down-local version of the inverse sensitivity mechanism. However, their method only applies to monotone functions; this restriction was removed by \citet{LRSS25}. Our results are based on this approach.
In particular, combining Theorems \ref{thm:main-formal-adp} and \ref{thm:main-formal-pdp} and setting $m=t$ recovers Theorem 3.1 of \citet{LRSS25}; our results are thus an extension of theirs to the case where $m>t$.

\citet{kohli2023differential} present an algorithm (which they call \textbf{TAHOE}) that is, roughly, a down-local version of propose-test-release.

\textbf{Sample-and-aggregate} \cite{NRS07} is closely related to our approach.
In its simplest form,\footnote{More generally, sample-and-aggregate allows overlapping subdatasets, but then the aggregator must handle the fact that a single person's data may affect multiple function values.} sample-and-aggregate partitions the dataset into smaller subdatasets, evaluates the function of interest on each subdataset, and then aggregates the function values in a differentially private manner (e.g., using smooth sensitivity applied to the aggregation function).
A single person's data will be in only one of the subdatasets and so a single person can only affect one of the values.
The advantage of the sample and aggregate approach is that it requires no structural knowledge about the function of interest and is computationally efficient. (Intuitively, it pushes the privacy analysis onto the aggregation function, rather than the function we want to evaluate.) The downside is that we evaluate the function on the smaller subdatasets. This can lead to significant loss in accuracy relative  to (non-privately) evaluating the function on the whole dataset.
Our algorithm addresses the downside of sample and aggregate by allowing us to evaluate the function on larger subdatasets, at the expense of requiring us to evaluate on more of these subdatasets.

\subsection{Lower Bounds}
\citet{LRSS25} prove lower bounds on both locality and query complexity, which are similar in spirit to our Theorem \ref{thm:lower}.
(Locality refers to $|x \setminus x'|$ where $x$ is the input and $x' \subseteq x$ is the subset on which the function is evaluated.) 
The main difference between our lower bound and their query complexity lower bound is the notion of accuracy. Theorem \ref{thm:lower} assumes a statistical accuracy guarantee -- that is, if $\pr{X \gets \mathcal{D}^{n-m}}{f(X)=\nu} \ge 0.999$, then $\pr{X \gets \mathcal{D}^n}{|M^f(X)-\nu|\le 1} \ge \frac12$ for arbitrary $\mathcal{D}$. 
In contrast, they \cite[Theorem 6.1]{LRSS25} assume an accuracy guarantee of the form $\pr{}{|M(x)-f(x)| \le \alpha} \ge 1-\beta$ for arbitrary $x$ under the promise that $f$ is Lipschitz.
Thus these results are formally incomparable.

\citet{awasthi2015limitations} prove lower bounds on query complexity for local Lipschitz filters. This lower bound was extended by \citet{lange2025local} to allow a nonzero failure probability. Local Lipschitz filters can be used to construct differentially private algorithms, but there is no converse -- i.e., these lower bounds do not immediately translate into a lower bound for differentially private algorithms.

\subsection{Subsequent Work}
\citet{brown2026privately} present a \emph{polynomial-time} algorithm for privately estimating \emph{monotone} statistics. This quantitatively improves on our work, but is qualitatively different in that it assumes monotonicity. They also extend the lower bound to algorithms that are differentially private only under the assumption that the function $f$ is monotone.

\subsection{Miscellaneous}
Our algorithms rely on combinatorial objects known as covering designs. Combinatorial designs appear in many places. Notably, \citet{park2024exactly} used balanced incomplete block designs to develop minimax-optimal locally differentially private algorithms. Furthermore, \citet{gentle2025necessity} showed that these combinatorial designs are in fact necessary to achieve optimality.

\section{Preliminaries}

\subsection{Notation}

For a natural number $n \in \mathbb{N}$, we denote $[n]:=\{1,2,\cdots,n\}$.
We use $\log$ to denote the natural logarithm.
%We follow the convention that upper case variable names correspond to random variables, while lower case variable names are deterministic. 
%
Throughout, we will let $\mathcal{X}$ denote the set of possible input data points. 
Then $\mathcal{X}^n$ denotes tuples of length $n$ and $\mathcal{X}^* := \bigcup_{n \in \mathbb{N}} \mathcal{X}^n$ denotes tuples of arbitrary length.

We treat tuples as sets (and we use set notation), but we also maintain consistent indexing of the elements. 
This should be intuitive, but to be completely formal, below we define the set notation that we use on tuples.
The reader should skip this subsection and only refer back if there is any confusion.

We assume that there is a special ``null'' element $\bot \in \mathcal{X}$. Informally, $\bot$ represents a missing element when the tuple is viewed as a set. (And we assume that $\bot$ is not in the support of the data distribution $\mathcal{D}$.)
%We will (informally) treat tuples like sets, in which case $\bot$ represents a missing element.\footnote{The reason for this cumbersome tuple-set duality is that we need to work with subsets of the input while maintaining consistent indexing of input elements. Casual readers may ignore this formalism. }
%To be precise, f
For tuples $x,x' \in \mathcal{X}^n$, we define the following set notations:
\begin{enumerate}
    \item The size of $x$ is the number of non-null elements: \[|x| := |\{i \in [n] : x_i \ne \bot\}| .\label{eq:size}\]
    \item A subset corresponds to replacing elements with nulls:  \[x' \subseteq x ~~ \iff ~~ \forall i \in [n] ~ (x'_i = x_i \vee x'_i = \bot) .\]
    \item Intersections and differences are given by \[\forall i \in [n] ~~~~~ ( x \cap x' )_i = \left\{\begin{array}{cl} x_i & \text{ if } x_i = x'_i \\ \bot & \text{ if } x_i \ne x'_i \end{array}\right\}\] and \[\forall i \in [n] ~~~~~ ( x \setminus x' )_i = \left\{\begin{array}{cl} \bot & \text{ if } x_i = x'_i \\ x_i & \text{ if } x_i \ne x'_i \end{array}\right\}.\] We have $x \cap x' \subseteq x$, $x \cap x' \subseteq x'$, and $x \setminus x' \subseteq x$. Also, $|x \setminus x'| + |x' \setminus x| = |\{i \in [n] : x_i \ne x'_i = \bot \vee \bot = x_i \ne x'_i \}| + 2 |\{ i \in [n] : \bot \ne x_i \ne x'_i \ne \bot \}|$.
    \item Given a set of indices $S \subseteq [n]$, define $x_S \in \mathcal{X}^n$ by \[\forall i \in [n] ~~~~~ (x_S)_i = \left\{\begin{array}{cl} x_i & \text{ if } i \in S \\ \bot & \text{ if } i \notin S \end{array}\right\}.\] Note that $x_S \subseteq x$ for all $S$ and, assuming $|x|=n$, we have $|x_S|=|S|$.
\end{enumerate}
We work with functions $f : \mathcal{X}^* \to \mathcal{Y}$ and we assume that null values are equivalent to removing elements from the tuple entirely. That is, for all $n \in \mathbb{N}$, $x \in \mathcal{X}^n$, and $i \in [n+1]$, if $x' = (x_1, x_2, \cdots, x_{i-1}, \bot, x_i, x_{i+1}, \cdots, x_n) \in \mathcal{X}^{n+1}$, then $f(x) = f(x')$.

\subsection{Differential Privacy}

We say that $x,x'\in\mathcal{X}^n$ are \emph{neighbouring} if $|x \setminus x'|+|x' \setminus x| = 1$. Equivalently, $x,x'\in\mathcal{X}^n$ with $x \ne x'$ are \emph{neighbouring} if there exists $i \in [n]$ such that $x_i=\bot$ or $x'_i=\bot$ and, for all $j \in [n] \setminus \{i\}$, we have $x_j=x'_j$. Informally, neighbouring inputs differ by the addition or removal of one element, which corresponds to one person's data. 

\begin{definition}[Differential Privacy \citep{DMNS06,DKMMN06}]\label{def:dp}
    A randomised algorithm $M : \mathcal{X}^n \to \mathcal{Y}$ satisfies $(\varepsilon,\delta)$-differential privacy if, for all neighbouring $x,x' \in \mathcal{X}^n$ and all measurable $V \subseteq \mathcal{Y}$,
    \[ \pr{}{M(x) \in V} \le e^\varepsilon \pr{}{M(x') \in V} + \delta .\]
\end{definition}

``Pure differential privacy'' (or ``pointwise differential privacy'') refers to the setting where $\delta=0$. In contrast ``approximate differential privacy'' refers to the setting where $\delta>0$.

Differential privacy satisfies many useful properties. One is postprocessing -- applying an arbitrary function to the output of a differentially private algorithm still results in a differentially private output, with no loss in parameters.
The other property we use is group privacy:

\begin{lemma}[Group privacy]\label{lem:group-privacy}
    Suppose $M : \mathcal{X}^n \to \mathcal{Y}$ is $(\varepsilon,\delta)$-differentially private.
    Suppose $x,x'\in\mathcal{X}^n$ are distance $t = |x \setminus x'|+|x' \setminus x|$ apart.
    Then, for all measurable $V \subseteq \mathcal{Y}$,
    \[ \pr{}{M(x) \in V} \le e^{t\varepsilon} \pr{}{M(x') \in V} + \frac{e^{t\varepsilon}-1}{e^\varepsilon-1}\delta .\]
\end{lemma}

Note that we do not allow \emph{replacement} of one person's data between neighbours -- this would instead be group privacy at distance $t=2$.

Since we define differential privacy with respect to addition or removal of one person's data, rather than replacement, the true size of the dataset $|x|$ is itself sensitive.
We avoid this annoying technicality by having a separate parameter $n$ (which should be thought of as the public size of the dataset, rather than its true size) that is provided to the algorithm separately from the sensitive dataset.\footnote{Alternatively, we can expend some of the privacy budget to compute a public size $n$ from the sensitive dataset.} If the true size of the sensitive dataset is larger than the public size $n$, we can discard elements. If the true size of the sensitive dataset is smaller than the public size $n$, we pad with nulls ($\bot$).

\subsection{Shifted Inverse Mechanism}\label{sec:shi}

The basis of our algorithm is the shifted inverse mechanism of \citet{FangDY22}. For completeness, we review this algorithm in Appendix \ref{app:shi}.

\begin{theorem}[Shifted Inverse Mechanism -- Pure DP]\label{thm:shi-pdp}
    Let $g : \mathcal{X}^* \to \mathcal{Y}$ be monotone -- i.e., $x' \subseteq x ~\implies~ g(x') \le g(x)$ -- where $\mathcal{Y} \subseteq \mathbb{R}$ is finite.
    Let $\varepsilon,\beta>0$.
    Then there exists a $(\varepsilon,0)$-differentially private $M : \mathcal{X}^* \to \mathcal{Y}$ such that, for all $x \in \mathcal{X}^*$, we have
    \[
        \pr{M}{g(x) \ge M(x) \ge \min \left\{ g(x') : x' \subseteq x, |x'| \ge |x| - t \right\} } \ge 1-\beta,
    \]
    where $t = 2\left\lceil \frac{2}{\varepsilon} \log\left(\frac{|\mathcal{Y}|}{\beta}\right) \right\rceil$. Furthermore, $M(x)$ only depends on the values $g(x')$ for $x' \subseteq x$.
\end{theorem}

We use a variant of the algorithm satisfying approximate differential privacy \citep{LRSS25,DPorg-down-sensitivity}:

\begin{theorem}[Shifted Inverse Mechanism -- Approx DP]\label{thm:shi-adp}
    Let $g : \mathcal{X}^* \to \mathcal{Y}$ be monotone -- i.e., $x' \subseteq x ~\implies~ g(x') \le g(x)$ -- where $\mathcal{Y} \subseteq \mathbb{R}$ is finite.
    Let $\varepsilon,\delta>0$.
    Then there exists a $(\varepsilon,\delta)$-differentially private $M : \mathcal{X}^* \to \mathcal{Y}$ such that, for all $x \in \mathcal{X}^*$, we have
    \[
        \pr{M}{g(x) \ge M(x) \ge \min \left\{ g(x') : x' \subseteq x, |x'| \ge |x| - t \right\} } = 1,
    \]
    where $t=\frac{1}{\varepsilon} \log(1/\delta) \exp(O(\log^* |\mathcal{Y}|))$ and $\log^*$ denotes the iterated logarithm.\footnote{The iterated logarithm is an extremely slow-growing function. It is the inverse of the exponential tower function, which satisfies the recurrence $\mathsf{tower}(n+1) = 2^{\mathsf{tower}(n)}$.} Furthermore, $M(x)$ only depends on the values $g(x')$ for $x' \subseteq x$.
\end{theorem}

We also consider a variant that satisfies Concentrated Differential Privacy \cite{dwork2016concentrated,bun2016concentrated} or Gaussian Differential Privacy \cite{dong2022gaussian}:

\begin{theorem}[Shifted Inverse Mechanism -- zCDP/GDP]\label{thm:shi-cdp}
    Let $g : \mathcal{X}^* \to \mathcal{Y}$ be monotone -- i.e., $x' \subseteq x ~\implies~ g(x') \le g(x)$ -- where $\mathcal{Y} \subseteq \mathbb{R}$ is finite.
    Let $\rho,\beta>0$.
    Then there exists $M : \mathcal{X}^* \to \mathcal{Y}$ satisfying $\rho$-zCDP and $\sqrt{2\rho}$-GDP such that, for all $x \in \mathcal{X}^*$, we have
    \[
        \pr{M}{g(x) \ge M(x) \ge \min \left\{ g(x') : x' \subseteq x, |x'| \ge |x| - t \right\} } \ge 1-\beta,
    \]
    where $t=O(\sqrt{\log(|\mathcal{Y}|/\beta)/\rho})$. Furthermore, $M(x)$ only depends on the values $g(x')$ for $x' \subseteq x$.
\end{theorem}

To the best of our knowledge, Theorem \ref{thm:shi-cdp} is novel (although it follows from known techniques); thus we discuss it in more detail in Appendix \ref{app:shi}.

\subsection{Covering Designs}\label{sec:designs}

Our algorithm also depends on a combinatorial object which is known as a covering design. 

\begin{definition}[Covering Design]\label{def:covering_design}
    Given $n,m,t\in\mathbb{N}$, $t\le m\le n$, a \emph{$(n,m,t)$-covering design} of size $k$ is a collection of sets $S_1, S_2, \cdots, S_k \subseteq [n]$ each of size $|S_i|=m$ with the property that, for every $T \subseteq [n]$ of size $|T| \le t$, there exists $i \in [k]$ such that $T \subseteq S_i$.
    We let $C(n,m,t)$ denote the smallest $k$ for which a $(n,m,t)$-covering design of size $k$ exists.
    %And $C(n,m,t)=\{S_1,S_2,\cdots,S_k\}$ denotes some $(n,m,t)$-covering design of minimal size.
\end{definition}

Covering designs are equivalent to what is known as Tur\'an systems \cite{Sidorenko95}.
To be precise, if $S_1, S_2, \cdots, S_k$ is a $(n,m,t)$-covering design, then $[n] \setminus S_1, [n] \setminus S_2, \cdots, [n] \setminus S_k$ is a $(n, n-t, n-m)$-Tur\'an system and vice versa.
Such a Tur\'an system has the property that for every $T \subseteq [n]$ of size $|T| = n-t$, there exists $i \in [k]$ such that $T \supseteq [n]\setminus S_i$.

In general, we do not have optimal constructions or even existential results for covering designs.
However, the following result gives reasonable bounds. 

\begin{proposition}\label{prop:covering-size}
    For all $n,m,t \in \mathbb{N}$ with $n \ge m \ge t$ we have
    \[\frac{{n \choose t}}{{m \choose t}} \le C(n,m,t) < \frac{{n \choose t}}{{m \choose t}} \left(1 + \log{ m \choose t} \right) + 1.\]
\end{proposition}
The lower bound is due to \citet{Schonheim64} and the upper bound is due to \textcite[Theorem 13.4]{ErdosS74}. Both proofs rely on the probabilistic method.
    
We also state (looser) bounds without binomial coefficients:
\begin{corollary}\label{cor:covering-size}
    For all $n,m,t \in \mathbb{N}$ with $n \ge m \ge t>1$ we have
    \[
        \left(\frac{n}{m}\right)^t \le C(n,m,t) <  \left(\frac{n e}{m}\right)^t \cdot\min\{ 1 + m \log 2, 1 + t \log m \} +1.
    \]
\end{corollary}

For completeness, we provide proofs of Proposition \ref{prop:covering-size} and Corollary \ref{cor:covering-size} in Appendix \ref{app:coveringdesigns}. 

Roughly speaking, the lower bound in Proposition \ref{prop:covering-size} is tighter than the upper bound. 
When the lower bound is exactly tight, the covering design is known as a Steiner system (which satisfies the stricter property that each $T \subseteq [n]$ of size $|T|=t$ is contained in \emph{exactly} one $S_i$, rather than \emph{at least} one). Results on the existence of Steiner systems \citep{keevash2014existence,keevash2024short} imply that the lower bound is exactly tight infinitely often. More specifically, for any fixed integers $m \ge t \ge 1$, there exists $n_0$ such that $C(n,m,t) = \frac{{n \choose t}}{{m \choose t}}$ for all $n \ge n_0$ satisfying $\frac{{n-i+1 \choose t-i+1}}{{m-i+1 \choose t-i+1}} \in \mathbb{Z}$ for all $i \in [t]$.
More generally, for any fixed integers $m \ge t \ge 1$, we have \cite{RODL198569} \[C(n,m,t) \le (1+o(1)) \frac{{n \choose t}}{{m \choose t}} ~~~\text{ as } n \to \infty.\] 
In particular, the lower bound is tight at the extreme choices of $m$: When $m=n$, we have $C(n,n,t)=1$. And, when $m=t$, we have \[ C(n,t,t) = {n \choose t}.\label{eq:m=t}\]
The setting where $m,t$ are fixed and $n \to \infty$ is of interest for our work. However, we are more interested in the setting where the ratio $n/m$ is constant and $n,m \to \infty$.
A simple result that helps grapple with this setting is \[\forall \ell,n,m,t \in \mathbb{N} ~~~~ C(\ell n,\ell m,t) \le C(n,m,t). \label{eq:chunks}\] Equation \ref{eq:chunks} follows by partitioning $[\ell n]$ into $n$ chunks of size $\ell$ and taking a $(n,m,t)$-covering design and applying it to the chunks instead of individual points.
Combining Equations \ref{eq:m=t} and \ref{eq:chunks} gives \[C(n,\frac{t}{t+1}n,t) = C(\ell(t+1),\ell t, t) \le C(t+1,t,t) = {t+1 \choose t} = t+1, \label{eq:samp-agg-design}\] where $\ell = \frac{n}{t+1}$ is assumed to be an integer.
%\nocite{GORDON1996270}
Finally, note that we can largely ignore integer divisibility issues due to monotonicity -- i.e., $C(n-i,m+j,t) \le C(n,m,t)$ when $n \ge n-i \ge m+j \ge m \ge t \ge 0$.

\section{Our Algorithm}

\begin{algorithm}
    \caption{Differentially Private Black-box Estimator}\label{alg:estimate}
    \begin{algorithmic}
        \Procedure{Estimate}{$f:\mathcal{X}^n\to\mathcal{Y},x\in\mathcal{X}^n,\varepsilon>0,\delta>0,S_1,\cdots, S_k \subseteq [n]$}
            \State Let $t = \frac{1}{\varepsilon} \log(1/\delta) \exp(O(\log^* |\mathcal{Y}|))$ as in Theorem \ref{thm:shi-adp} be an integer. 
            \State Assert that $S_1, S_2, \cdots, S_k \subseteq [n]$ is a $(n,m,t)$-covering design (Definition \ref{def:covering_design}). \State\Comment{$n \ge m \ge t$, $k \ge C(n,m,t)$.}
            \State Compute $f(x_{[n]\setminus S_i})$ for each $i \in [k]$. \Comment{Only values of $f$ we need.}
            %\State Define $\widetilde{f}_m:\mathcal{X}^n\to\mathcal{Y}$ by \[\widetilde{f}_m(x') := \left\{\begin{array}{cl} f(x') & \text{ if } |x'| \ge n-m \\ -\infty & \text{ if } |x'| < n-m\end{array}\right\}.\]
            \State Define $g:\mathcal{X}^n\to\mathcal{Y}$ by
            \[g(x') := \max\{ f(x'_{[n] \setminus S_i}) : i \in [k], |x'_{[n] \setminus S_i}| = n-|S_i| \}.\label{eq:g}\]
            %\[g(x') := \max\{f(x_{[n] \setminus S_i}) : i \in [k], x_{[n] \setminus S_i} \subseteq x'\}.\] %where we define $\max \emptyset = -\infty$.
            \State \Comment{Define $\max \emptyset := \min \mathcal{Y}$.}
            \State Let $M$ be the Shifted Inverse Mechanism from Theorem \ref{thm:shi-adp} applied to $g$.
            \State \Return $M(x) \in \mathcal{Y}$.
        \EndProcedure
    \end{algorithmic}
\end{algorithm}

Our algorithm is specified in Algorithm \ref{alg:estimate}.
We can see from the algorithm description that the shifted inverse mechanism (\S\ref{sec:shi}) is the main ingredient.
In terms of the analysis, we must check three things:
\begin{itemize}
    \item[(i)] The function $g$ is monotone, as required for the shifted inverse mechanism.
    \item[(ii)] The $k$ evaluations of $f$ suffice for all computations.
    \item[(iii)] The accuracy guarantee of the shifted inverse mechanism translates to the desired accuracy of our algorithm.
\end{itemize}
We address these claims in the following two lemmata. 
But first we provide some intuition for our choice of the function $g$:

Ideally, we want $g=f$, but we need $g$ to be monotone.
A natural way to monotonise $f$ is to set $g(x') = \max \{ f(\check{x}) : \check{x} \subseteq x' \}$. For example, if $f(x) = \sum_i x_i$ is the sum, then the corresponding monotonisation $g$ would simply be the sum over positive terms $g(x) = \sum_i \max\{ x_i, 0 \}$.
The main issue with this monotonisation is that evaluating $g$ requires evaluating $f$ exponentially many times.
We fix this issue by only evaluating $f$ on carefully-chosen subsets of the input. This is where the covering design $S_1, \cdots, S_k$ enters the picture. Setting $g(x') = \max\{ f(x'_{[n] \setminus S_i}) : i \in [k] \}$ \emph{almost} works -- the subtlety is that the shifted inverse mechanism evaluates $g(x')$ for many $x' \subseteq x$. We add the restriction $|x'_{[n] \setminus S_i}| = n-|S_i|$ to address this subtlety and ensure that evaluating $f(x_{[n] \setminus S_i})$ for $i \in [k]$ suffices to compute $g(x')$ for all $x' \subseteq x$. 

Next: Why choose the subsets to form a covering design? Monotonicity (and therefore privacy) holds for any choice of subsets. And for oracle efficiency we just want to minimize the number of subsets $k$. The last requirement is statistical accuracy. For this we want each individual element in the maximum $f(x'_{[n] \setminus S_i})$ to be statistically accurate, which means we want $[n] \setminus S_i$ to be large, so we want $S_i$ to be small. Finally, the shifted inverse mechanism's accuracy guarantee requires $g$ to have low down-sensitivity.\footnote{We avoid the formalism of down sensitivity. To be precise, the $t$-down sensitivity of $f$ at $x$ is $\mathsf{DS}_f^t(x) := \max \{ |f(x) - f(x')| : x' \subseteq x, |x \setminus x'| \le t \}$.} This translates to the covering requirement -- each small $T \subseteq [n]$ must be contained in some $S_i$.

\begin{lemma}\label{lem:mon}
    The function $g$ defined by Equation \ref{eq:g} in Algorithm \ref{alg:estimate} is monotone -- i.e., $x' \subseteq x \implies g(x') \le g(x)$.
    Furthermore, all values $g(x')$ for $x' \subseteq x$ can be computed from the values $f(x_{[n] \setminus S_i})$ for $i \in [k]$.
\end{lemma}
\begin{proof}
    Recall $g:\mathcal{X}^n\to\mathcal{Y}$ is given by
    \[g(x') := \max\{ f(x'_{[n] \setminus S_i}) : i \in [k], |x'_{[n] \setminus S_i}| = n-|S_i| \},\tag{\ref{eq:g}}\] where we set $\max \emptyset = \min \mathcal{Y}$ to cover the corner case.
    %Recall that $|S_i|=m$ for all $i \in [k]$. 
    The requirement $|x'_{[n] \setminus S_i}| = n-|S_i|$ is equivalent to requiring that there are no null values in $x'_{[n] \setminus S_i}$ -- i.e., $\forall j \in [n] \setminus S_i ~~ x'_j \ne \bot$. (Recall how size is defined in Equation \ref{eq:size}.)
    
    Thus, if we remove an element from the input, then this removes from the maximum all function values that depend on the removed input.
    That is, if we remove, say, the $j$-th element $x'_j$ from the input $x'$ by replacing it with $\bot$, then this removes from the maximum all indices $i \in [k]$ such that $j \notin S_i$.
    Removing elements can only decrease the maximum, which implies the monotonicity of $g$.

    For $x' \subseteq x$ and $i \in [k]$, if $|x'_{[n] \setminus S_i}| = n-|S_i|$, then $x'_{[n] \setminus S_i} = x_{[n] \setminus S_i}$. Thus, if $x' \subseteq x$, then $g(x') = \max \{ f(x_{[n] \setminus S_i}) : i \in [k], |x'_{[n] \setminus S_i}|=n-|S_i|\}$. In other words, for any $x' \subseteq x$, we can compute the value $g(x')$ from the $k$ values $f(x_{[n] \setminus S_i})$ for $i \in [k]$, as required.
    (See Section \ref{sec:limitations} for further discussion about computing $g$ from $f$.)
\end{proof}

\begin{lemma}\label{lem:acc}
    Let $g$, $t$, and $S_1,\cdots,S_k$ be as in Algorithm \ref{alg:estimate}.
    Let $x,x' \in \mathcal{X}^n$ with $x' \subseteq x$ with $|x'| \ge n-t$ and $|x|=n$.
    Then \[\max\{ f(x_{[n] \setminus S_i}) : i \in [k] \} \ge g(x) ~~~\text{ and }~~~ g(x') \ge \min\{ f(x_{[n] \setminus S_i}) : i \in [k] \}.\]
\end{lemma}
\begin{proof}
    The first part of the claim is in fact an equality: $\max\{ f(x_{[n] \setminus S_i}) : i \in [k] \} = g(x)$ and follows from the definition of $g$ (Equation \ref{eq:g}) and assumption $|x|=n$, which implies $|x_{[n]\setminus S_i}|=n-|S_i|=n-m$ for all $i \in [k]$.
    The second part of the claim relies on the fact that $S_1, \cdots, S_k$ is a $(n,m,t)$-covering (Definition \ref{def:covering_design}).
    Let $T \subseteq [n]$ be such that $x' = x_{[n] \setminus T}$ and $|T| \le t$.
    Then, by definition, there exists some $i \in [k]$ with $T \subseteq S_i$.
    It follows that $x'_{[n] \setminus S_i}=x_{[n] \setminus S_i}$ and $|x'_{[n] \setminus S_i}|=|x_{[n] \setminus S_i}|=n-m$; hence, $g(x') = \max\{ f(x'_{[n] \setminus S_i}) : i \in [k], |x'_{[n] \setminus S_i}| = n-m \} \ge f(x_{[n]\setminus S_i})$, as required.
\end{proof}

Combining Lemmas \ref{lem:acc} and \ref{lem:mon} with the guarantee of the shifted inverse mechanism in Theorem \ref{thm:shi-adp} and an optimal covering design, gives us the following guarantee.

\begin{theorem}[Main Result -- General Version] \label{thm:main-formal-adp}
    Let $f : \mathcal{X}^n \to \mathcal{Y}$ with $\mathcal{Y} \subseteq \mathbb{R}$ finite. Let $\varepsilon,\delta>0$.
    Let $t = \frac{1}{\varepsilon} \log(1/\delta) \exp(O(\log^* |\mathcal{Y}|))$ as in Theorem \ref{thm:shi-adp} be an integer.
    Let $m \in \mathbb{N}$ satisfy $n \ge m \ge t$.
    Let $M : \mathcal{X}^n \to \mathcal{Y}$ be \textsc{Estimate} from Algorithm \ref{alg:estimate} instantiated with $f,\varepsilon,\delta$ and a $(n,m,t)$-covering design $S_1,\cdots,S_k\subseteq [n]$ of size $k$. % -- i.e., $k=C(n,m,t)$.
    Then we have the following properties.
    \begin{itemize}
        \item \textbf{Privacy}: $M$ is $(\varepsilon,\delta)$-differentially private.
        \item \textbf{Accuracy}: For any input $x \in \mathcal{X}^n$ of size $|x|=n$,\footnote{Note that the accuracy guarantee holds with probability 1. This is only possible with approximate differential privacy (i.e., $\delta>0$) \cite{DPorg-fail-prob}. In contrast, Theorems \ref{thm:main-formal-pdp} and \ref{thm:main-formal-cdp} have a nonzero failure probability, since these are for pure and concentrated differential privacy respectively.} \[\max\{f(x_{[n]\setminus S_i}) : i \in [k]\} \ge M(x) \ge \min\{f(x_{[n]\setminus S_i}) : i \in [k]\}.\]
        \item \textbf{Oracle Efficiency}: $M(x)$ only depends on the $k$ values $f(x_{[n] \setminus S_i})$ for $i \in [k]$.
    \end{itemize}
\end{theorem}
\begin{proof}
    Privacy follows from the privacy guarantee of the shifted inverse mechanism (Theorem \ref{thm:shi-adp}) and postprocessing; this requires $g$ to be monotone, which is guaranteed by the first part of Lemma \ref{lem:mon}.
    Efficiency follows from the second part of Lemma \ref{lem:mon} -- for all $x' \subseteq x$, $g(x')$ is determined by the $k$ values $f(x_{[n]\setminus S_i})$ for $i \in [k]$ and the shifted inverse mechanism only accesses the input by evaluating $g(x')$ with $x' \subseteq x$.
    The accuracy guarantee of the shifted inverse mechanism (Theorem \ref{thm:shi-adp}) is that \[g(x) \ge M(x) \ge \min \left\{ g(x') : x' \subseteq x, |x'| \ge |x| - t \right\}.\]
    By Lemma \ref{lem:acc}, $\max\{ f(x_{[n] \setminus S_i}) : i \in [k] \} \ge g(x)$ and \[ \min \left\{ g(x') : x' \subseteq x, |x'| \ge |x| - t \right\} \ge \min\{ f(x_{[n] \setminus S_i}) : i \in [k] \}.\]
    Combining the bounds yields the accuracy guarantee and completes the proof.
\end{proof}

Theorem \ref{thm:main} in the introduction is a simplification of Theorem \ref{thm:main-formal-adp}.
\begin{proof}[Proof of Theorem \ref{thm:main}.]
    The algorithm $M^f$ promised by Theorem \ref{thm:main} is \textsc{Estimate} from Algorithm \ref{alg:estimate} instantiated with an optimal covering design i.e. $k = C(n,m,t)$.
    The bounds on $k$ in Theorem \ref{thm:main} follow from Proposition \ref{prop:covering-size} and Corollary \ref{cor:covering-size}.
    The privacy and oracle efficiency guarantees of Theorem \ref{thm:main} follow immediately from those of Theorem \ref{thm:main-formal-adp}.
    It only remains to translate the accuracy guarantee:
    Suppose we have an input $X \in \mathcal{X}^n$ of size $|X|=n$ that consists of $n$ independent samples from an (unknown) distribution $\mathcal{D}$.
    For each $i \in [k]$, the subset of the input $X_{[n]\setminus S_i}$ corresponds to $n-m$ independent samples from $\mathcal{D}$, since $|S_i|=m$.
    Theorem \ref{thm:main} assumes that $\pr{X \gets \mathcal{D}^{n-m}}{|f(X)-\nu|\le\alpha}\ge 1-\beta$ for some value $\nu$ (where $\nu$ depends on $\mathcal{D}$).
    Thus, by a union bound, $\pr{X \gets \mathcal{D}^n}{\forall i \in [k] ~~ |f(X_{[n]\setminus S_i})-\nu|\le\alpha} \ge 1-k\beta$.
    From Theorem \ref{thm:main-formal-adp}, we have \[\max\{f(X_{[n]\setminus S_i}) : i \in [k]\} \ge M(X) \ge \min\{f(X_{[n]\setminus S_i}) : i \in [k]\}.\] % where $S_i \subseteq [n]$ with $|S_i|=m$ for all $i \in [k]$. (This part holds with probability 1.)
    It follows that $\pr{X \gets \mathcal{D}^n}{|M(X)-\nu|\le\alpha} \ge 1-k\beta$, as required.
\end{proof}

Note that the proof of Theorem \ref{thm:main} uses a union bound over the $k$ evaluations of the function $f$. \label{union-bound-discussion}
This bound may be overly pessimistic, since the events are not independent.
As a simple thought experiment, suppose each data point is corrupted with probability $\gamma$ and the function evaluation fails only if one of the data points is corrupted. In this case, it suffices to take a union bound over the $n$ data points, rather than over the $k$ evaluations.
Theorem \ref{thm:main-formal-adp} is stated differently than Theorem \ref{thm:main} so that we can use a different analysis, rather than a union bound.

\subsection{Pure \& Concentrated DP Variants}
Theorem \ref{thm:main-formal-adp} is stated for approximate differential privacy. We also state results for the pure and concentrated variants of differential privacy. These results follow by applying the relevant versions of the shifted inverse mechanism (\S\ref{sec:shi}).
However, this requires us to introduce an added failure probability in the mechanism.

\begin{theorem}[Main Result -- Pure DP Version] \label{thm:main-formal-pdp}
    Let $f : \mathcal{X}^n \to \mathcal{Y}$ with $\mathcal{Y} \subseteq \mathbb{R}$ finite. Let $\varepsilon,\beta>0$.
    Let $t = 2\left\lceil \frac{2}{\varepsilon} \log\left(\frac{|\mathcal{Y}|}{\beta}\right) \right\rceil$ as in Theorem \ref{thm:shi-pdp}.
    Let $m \in \mathbb{N}$ satisfy $n \ge m \ge t$ and let $k \ge C(n,m,t)$.
    Then there exists $M : \mathcal{X}^n \to \mathcal{Y}$ with the following properties.
    \begin{itemize}
        \item \textbf{Privacy}: $M$ is $(\varepsilon,0)$-differentially private.
        \item \textbf{Accuracy}: For any input $x \in \mathcal{X}^n$ of size $|x|=n$, \[\pr{M}{\max\{f(x_{[n]\setminus S_i}) : i \in [k]\} \ge M(x) \ge \min\{f(x_{[n]\setminus S_i}) : i \in [k]\}} \ge 1-\beta.\]
        \item \textbf{Oracle Efficiency}: $M(x)$ only depends on the $k$ values $f(x_{[n] \setminus S_i})$ for $i \in [k]$.
    \end{itemize}
\end{theorem}

\begin{theorem}[Main Result -- Concentrated DP Version] \label{thm:main-formal-cdp}
    Let $f : \mathcal{X}^n \to \mathcal{Y}$ with $\mathcal{Y} \subseteq \mathbb{R}$ finite. Let $\rho,\beta>0$.
    Let $t = O\left( \sqrt{ \frac{1}{\rho} \log\left(\frac{|\mathcal{Y}|}{\beta}\right) } \right)$ as in Theorem \ref{thm:shi-cdp}.
    Let $m \in \mathbb{N}$ satisfy $n \ge m \ge t$ and let $k \ge C(n,m,t)$.
    Then there exists $M : \mathcal{X}^n \to \mathcal{Y}$ with the following properties.
    \begin{itemize}
        \item \textbf{Privacy}: $M$ is $\rho$-zCDP and $\sqrt{2\rho}$-GDP.
        \item \textbf{Accuracy}: For any input $x \in \mathcal{X}^n$ of size $|x|=n$, \[\pr{M}{\max\{f(x_{[n]\setminus S_i}) : i \in [k]\} \ge M(x) \ge \min\{f(x_{[n]\setminus S_i}) : i \in [k]\}} \ge 1-\beta.\]
        \item \textbf{Oracle Efficiency}: $M(x)$ only depends on the $k$ values $f(x_{[n] \setminus S_i})$ for $i \in [k]$.
    \end{itemize}
\end{theorem}

To guarantee $(\varepsilon,\delta)$-differential privacy, it suffices \cite[Remark 15]{steinke2022composition} to have $\rho$-zCDP with \[\rho = \frac{\varepsilon^2}{4\log(1/\delta) + 4\varepsilon}.\] 
Substituting this bound into Theorem \ref{thm:main-formal-cdp} gives \[t = O\left( \frac{1}{\varepsilon} \sqrt{ (\log(1/\delta)+\varepsilon) \cdot \log\left(\frac{|\mathcal{Y}|}{\beta}\right) } \right).\]
We can compare this bound to $t = \frac{1}{\varepsilon} \log(1/\delta) \exp(O(\log^* |\mathcal{Y}|))$ in Theorem \ref{thm:main-formal-adp}.
In particular, if the output space $\mathcal{Y}$ is not too large and we can tolerate a reasonable failure probability $\beta$, then $\sqrt{\log(|\mathcal{Y}|/\beta)}$ is not much larger than $\exp(O(\log^*|\mathcal{Y}|))$, which means that the dominant difference between the bounds is that Theorem \ref{thm:main-formal-adp} has a $\log(1/\delta)$ term, where Theorem \ref{thm:main-formal-cdp} gives $\sqrt{\log(1/\delta)}$. That is to say, in a reasonable parameter regime, Theorem \ref{thm:main-formal-cdp} is better than Theorem \ref{thm:main-formal-adp}. (This comparison, of course, depends on the constants hidden by the big-O notation.)

Furthermore, if the desired privacy failure probability $\delta$ is sufficiently small, then the bound of $t = O\left(\frac{1}{\varepsilon} \log(|\mathcal{Y}|/\beta)\right)$ from Theorem \ref{thm:main-formal-pdp} -- which is independent of $\delta$ -- may dominate the bounds from Theorems \ref{thm:main-formal-adp} and \ref{thm:main-formal-cdp}.

\subsection{Example Applications}

We now consider some example applications of our algorithm. We consider the mean and the maximum. Obviously, there is no need to treat these simple functions as black boxes and tailored algorithms would be more efficient than our method. However, it is instructive to see how well our method works for these examples where we know what outcome to expect. 
We also consider learning parities, which gives an asymptotic improvement over sample-and-aggregate.

\textbf{Gaussian Mean Estimation:}
Differentially private mean estimation is well-studied \citep{barber2014privacy,karwa2018finite,bun2019average,kamath2019privately,biswas2020coinpress,kamath2020private,WangXDX20,DuFMBG20,huang2021instance,CaiWZ21,brown2021covariance,liu2021robust,hopkins2022efficient,liu2022differential,KamathLZ22,kothari2022private,tsfadia2022friendlycore,alabi2023privately,KuditipudiDH23,asi2023robustness,BrownHS23,hopkins2023robustness,zhao2024huber,aumuller2024plan,singhal2024polynomial,dagan2024dimension,chaudhuri2024mean,yu2024gaussian,kamath2025bias,agarwal2025private,zampetakis2025private,dong2026tight}.
We can apply our algorithm to the function $f : (\mathbb{R} \cup \{\bot\})^* \to (\mathbb{R} \cup \{\bot\})$ given by $f(x) = \frac{\sum_{i : x_i \ne \bot} x_i}{\sum_{i : x_i \ne \bot} 1}$ and compare with known results. %(with $f(x)=\bot$ if the former definition results in $\frac00$).

The accuracy achievable for private mean estimation depends on the data distribution \cite{kamath2020private}. Namely, it depends on how heavy-tailed the data distribution is. %Following \citet{kamath2020private}, we assume that the data distribution $\mathcal{D}$ has an unknown mean $\mu = \ex{X \gets \mathcal{D}}{X}$ and a moment bound of the form $\ex{X \gets \mathcal{D}}{|X-\mu|^p} \le 1$ for some $p \ge 1$.
For simplicity we consider Gaussian data $\mathcal{D} = \mathcal{N}(\mu,\sigma^2)$. This has very light tails, which makes it amenable to simple differentially private algorithms. % and heavy tailed data where we know only a bound on the variance $\ex{X \gets \mathcal{D}}{(X-\mu)^2} \le 1$, where the mean is $\mu = \ex{X \gets \mathcal{D}}{X}$.

\begin{corollary}[Theorem \ref{thm:main} applied to Gaussian mean estimation]\label{cor:gauss-mean-est}
    Let $\mathcal{Y} \subseteq \mathbb{R}$ be finite. 
    Let $\varepsilon,\delta,\hat\beta>0$ and let $t=\frac1\varepsilon \log(1/\delta) \exp(O(\log^* |\mathcal{Y}|))$ as in Theorem \ref{thm:main}.
    Let $n \ge m \ge t$. 
    Then there exists a $(\varepsilon,\delta)$-differentially private $M : \mathbb{R}^* \to \mathcal{Y}$ with the following property.
    Let $\mu \in \mathbb{R}$ and $\sigma>0$. If $X \gets \mathcal{N}(\mu,\sigma^2)^n$ consists of $n$ i.i.d.~samples from a Gaussian with mean $\mu$ and variance $\sigma^2$, then \[\pr{X \gets \mathcal{N}(\mu,\sigma^2)^n}{|M(X)-\mu| \le O\left(\sigma\sqrt{\frac{t \log (n/m) + t + \log(\log(m)/\hat\beta)}{n-m}}\right)+\tilde\mu_{\mathcal{Y}} } \ge 1-\hat\beta,\label{eq:gauss-mean-acc}\]
    where $\tilde\mu_{\mathcal{Y}} := \min_{\nu \in \mathcal{Y}} |\mu-\nu|$ is the discretization error from rounding the output to $\mathcal{Y}$.
\end{corollary}
Setting $m=n/2$ in Equation \ref{eq:gauss-mean-acc} gives \[\pr{X \gets \mathcal{N}(\mu,\sigma^2)^n}{|M(X)-\mu| \le O\left(\sigma\sqrt{\frac{\frac1\varepsilon\log(1/\delta)\exp(O(\log^*|\mathcal{Y}|)) + \log(\log(n)/\hat\beta)}{n}}\right)+\tilde\mu_{\mathcal{Y}} } \ge 1-\hat\beta.\label{eq:gauss-mean-sub}\]
The accuracy guarantee of Corollary \ref{cor:gauss-mean-est} is, unfortunately, suboptimal. Namely, ignoring the other parameters, the optimal accuracy bound \cite{bun2019average} scales as $\alpha 
\le O\left(\frac{\sigma}{\sqrt{n}} + \frac{\sigma}{\varepsilon n}\right)$, whereas Equation \ref{eq:gauss-mean-sub} gives $\alpha 
\le O\left(\frac{\sigma}{\sqrt{\varepsilon n}}\right)$.
\begin{proof}
    Define $f : (\mathbb{R} \cup \{\bot\})^* \to \mathcal{Y}$ to be the closest point in $\mathcal{Y}$ to the mean of the dataset. In symbols, $f(x) = \argmin_{\nu \in \mathcal{Y}} \left|\nu-\frac{\sum_{i : x_i \ne \bot} x_i}{\sum_{i : x_i \ne \bot} 1}\right|$ (with $f(x)$ being defined to take an arbitrary fixed value in $\mathcal{Y}$ in the case where $\sum_{i : x_i \ne \bot} 1 = 0$).
    Now we apply Theorem \ref{thm:main} to this function to obtain the mechanism $M=M^f$.
    We immediately have the required differential privacy guarantee. It only remains to analyze the accuracy.
    Theorem \ref{thm:main} promises us that \[\text{ if } \pr{X \gets \mathcal{D}^{n-m}}{|f(X)-\nu|\le\alpha}\ge 1-\beta, \text{ then } \pr{X \gets \mathcal{D}^n}{|M^f(X)-\nu|\le\alpha}\ge 1- k\beta,\label{eq:thm1}\] where $k=C(n,m,t) \le O\left( \left(\frac{ne}{m}\right)^t \cdot t \log m\right)$.
    We know the distribution of the mean of $n-m$ independent samples from $\mathcal{D}=\mathcal{N}(\mu,\sigma^2)$: For all $\beta>0$, we have \[\pr{X \gets \mathcal{D}^{n-m}}{\left|\overline{X} - \mu\right| \le \sigma\sqrt{\frac{2\log(2/\beta)}{n-m}}} \ge 1-\beta \label{eq:gauss-mean},\] where $\overline{X} := \frac{1}{n-m}\sum_i X_i$. Hence \[\pr{X \gets \mathcal{D}^{n-m}}{|f(X) - \mu| \le 2\sigma\sqrt{\frac{2\log(2/\beta)}{n-m}}+\tilde\mu_{\mathcal{Y}}} \ge 1-\beta. \label{eq:gauss-mean-round}\]
    Equation \ref{eq:gauss-mean-round} has a multiplicative factor of $2$ and an additive factor of $\tilde\mu_{\mathcal{Y}}$ compared to the plain Gaussian tail bound in Equation \ref{eq:gauss-mean}. This is because $f$ rounds its output to $\mathcal{Y}$. In particular, if $\mu^*_{\mathcal{Y}} = \argmin_{\nu \in \mathcal{Y}} |\nu-\mu|$, we have \[|f(X)-\mu| \le |f(X) - \overline{X}| + |\overline{X}-\mu| \le |\mu^*_{\mathcal{Y}}-\overline{X}| + |\overline{X}-\mu| \le |\mu^*_{\mathcal{Y}}-\mu| + 2|\overline{X}-\mu| = 2|\overline{X}-\mu|+\tilde\mu_{\mathcal{Y}}.\] The first and last inequalities follow from the triangle inequality and the inequality $|f(X) - \overline{X}| \le |\mu^*_{\mathcal{Y}}-\overline{X}|$ follows from the fact that, by definition, $f(X)$ is the closest point in $\mathcal{Y}$ to $\overline{X}$ and the fact that $\mu^*_{\mathcal{Y}} \in \mathcal{Y}$.

    Now we substitute Equation \ref{eq:gauss-mean-round} into Equation \ref{eq:thm1} and set $\beta = \hat\beta/k$, so that the final accuracy guarantee has the desired failure probability $\hat\beta = k\beta$.
    This gives the accuracy parameter \[\alpha = 2\sigma\sqrt{\frac{2\log(2k/\hat\beta)}{n-m}}+\tilde\mu_{\mathcal{Y}} = O\left(\sigma\sqrt{\frac{t \log (n/m) + t + \log(\log(m)/\hat\beta)}{n-m}}\right)+\tilde\mu_{\mathcal{Y}},\] as required.
\end{proof}

% Code for plots:
% https://colab.corp.google.com/drive/1CvLgaOwWwekNrjXRPdHIanJZXxEYEAbd?usp=sharing

\begin{figure}
    \centering
    \includegraphics[width=0.5\linewidth]{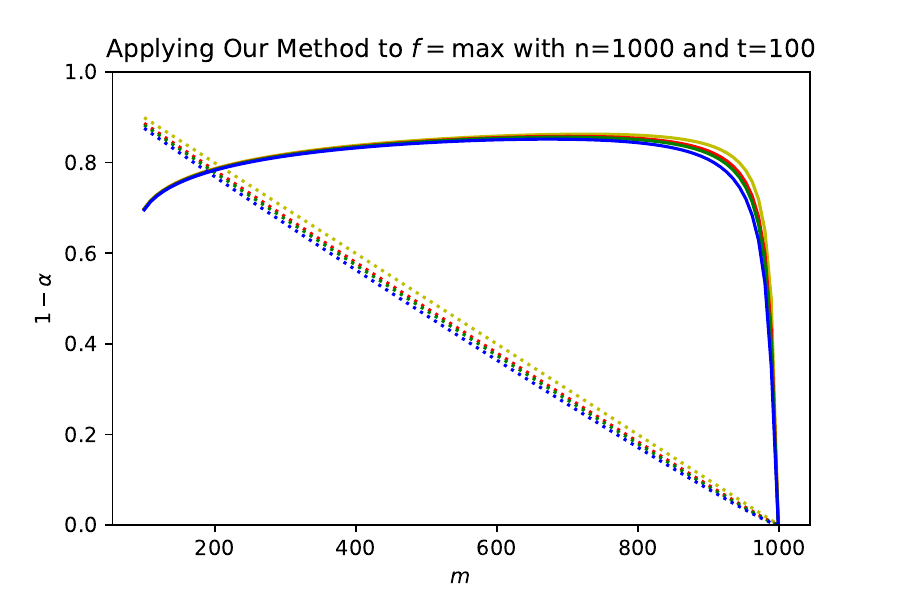}%
    \includegraphics[width=0.5\linewidth]{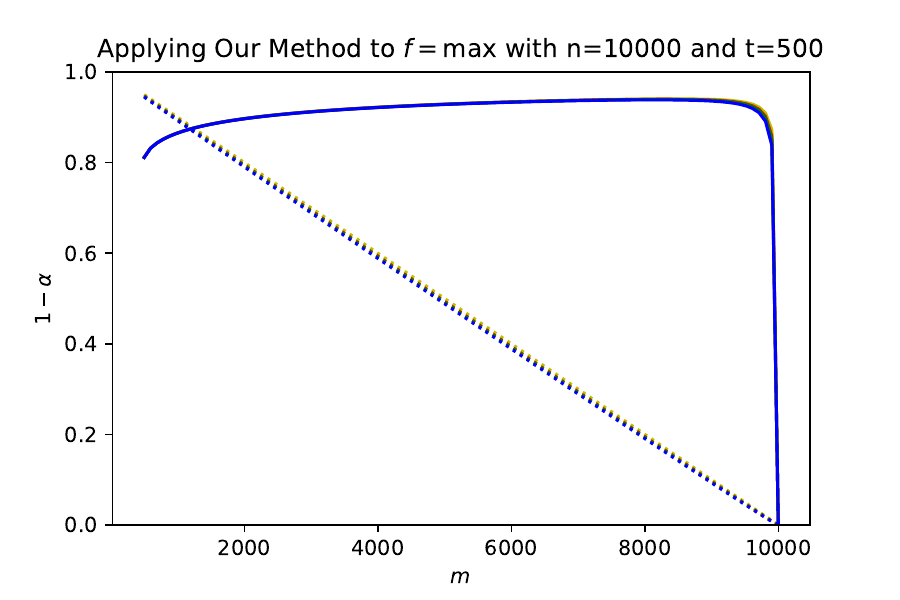}
    \caption{
    Accuracy bounds for applying Algorithm \ref{alg:estimate} to $f=\max$ with data $X$ that is i.i.d.~uniform in $[0,1]$. The bounds shown are lower bounds and the ``ideal'' answer is $1$, so higher is better.
    The vertical axis $1-\alpha$ represents the output value of our algorithm -- i.e., $\pr{X \gets [0,1]^n}{M^f(X) \ge 1-\alpha} \ge 1-\beta$.
    The parameter $m$ on the horizontal axis represents the number of elements discarded for privacy and ranges from $t$ to $n-1$. 
    Solid lines represent the union bound given in Theorem \ref{thm:main} and dotted lines represent an alternative bound using the order statistics.
    \textcolor{blue}{Blue} lines (lowest) represent a 99\% confidence bound (i.e., $\beta=0.01$); \textcolor{green}{green} lines are 95\% confidence ($\beta=0.05$); \textcolor{red}{red} lines are 90\% confidence ($\beta=0.1$); and \textcolor{yellow}{yellow} lines (highest) are 50\% confidence (i.e., median / $\beta=0.5$).
    }
    \label{fig:max-bound}
\end{figure}

\textbf{Maximum:} 
Consider the function $f : (\mathbb{R} \cup \{\bot\})^* \to (\mathbb{R} \cup \{\bot\})$ given by $f(x) = \max\{ x_i : x_i \ne \bot \}$ (where $\max\emptyset=\bot$).\footnote{In other words, $\bot=-\infty$.} This function has infinite sensitivity and even infinite local sensitivity, so it's a tough example to apply our algorithm to. 

For simplicity, we will ignore the need for the output space $\mathcal{Y}$ to be finite.\footnote{ In Figure \ref{fig:max-bound}, we simply set the parameter $t$ to some fixed value and thus ignore the fact that $t$ depends on the size of the output space $\mathcal{Y}$.} Of course, we could discretize the output space, but that would increase the complexity of the analysis without adding any insight.

Now consider data $x$ that is drawn i.i.d.~from the uniform distribution on $[0,1]$.
In this case, the ``ideal'' output is $1$, since that is what the output converges to as the sample size increases.
Thus the measure of accuracy for our algorithm is how close the output is to $1$.

For a sample of size $n-m$, the maximum follows a $\mathsf{Beta}(n-m,1)$ distribution. This has mean $ 1 - \frac{1}{n-m+1}$. And the cumulative distribution function is given by \[\pr{X_1, \cdots, X_{n-m} \gets [0,1]}{\max\{X_1, \cdots, X_{n-m}\} \ge 1-\alpha } = 1 - (1-\alpha)^{n-m} = 1 - \beta.\]
We can plug this into Theorem \ref{thm:main} to get the statistical accuracy guarantee that \[\pr{X_1, \cdots, X_n \gets [0,1]}{M^f(X) \ge 1-\alpha} \ge 1-k\beta = 1-k(1-\alpha)^{n-m},\] where $k = C(n,m,t) < \frac{{n \choose t}}{{m \choose t}} \left(1 + \log{ m \choose t} \right) +1 \le O\left( \left(\frac{ne}{m}\right)^t \cdot t \log m\right)$.
We plot this bound in Figure \ref{fig:max-bound} (using $k = \left\lceil \frac{{n \choose t}}{{m \choose t}} \left(1 + \log{ m \choose t} \right) \right\rceil$).

As discussed on page \pageref{union-bound-discussion}, Theorem \ref{thm:main} applies a union bound and this bound may be loose.
Thus we also consider an alternative bound. Specifically, Theorem \ref{thm:main-formal-adp} gives a bound of the form $M(X) \ge \min\{f(X_{[n]\setminus S_i}) : i \in [k]\}$. Instead of bounding this with a union bound over $i \in [k]$, we can use the fact that each set $[n] \setminus S_i$ is of size $n-m$ and thus $f(X_{[n]\setminus S_i}) = \max\{ X_j : j \in [n] \setminus S_i\} \ge X_{(n-m)}$, where $X_{(1)} \le X_{(2)} \le X_{(3)} \le \cdots \le X_{(n)}$ denotes the order statistics of the sample $X$. In particular, $X_{(n-m)}$ follows a $\mathsf{Beta}(n-m,m+1)$ distribution with mean $1-\frac{m+1}{n+1}$. Figure \ref{fig:max-bound} also plots the bound $M(X) \ge X_{(n-m)}$. 

In Figure \ref{fig:max-bound} we see several things: 
First, the solid lines are relatively flat for most of the range, which shows that the bound of Theorem \ref{thm:main} is relatively insensitive to the choice of $m$ in this range. This is because two effects cancel out -- decreasing $m$ improves the accuracy of each subsample of size $n-m$, but it increases the number of subsamples $k$ that we must union bound over. 
However, as $m$ approaches $n$, we do see a significant drop in accuracy as the sample size here shrinks.
Second, the union bound is indeed not tight, as illustrated by the dotted lines giving a better bound for small values of $m$. Small $m$ is the parameter regime where the subsamples are highly overlapping and thus far from independent.
Third, the four different coloured lines, which represent different probability bounds, are very close. This indicates that the outcomes (or, at least, our bounds on those outcomes) are highly concentrated.

\textbf{Learning Parities:}
Consider a problem that can be solved (non-privately) with $d$ samples, but not fewer. Specifically, we consider PAC learning parities: each sample is a pair $(x,y)$ where $x \in \{0,1\}^d$ is uniformly random and $y = \langle h, x \rangle \mod 2$ for some $h \in \{0,1\}^d$.  For simplicity, the goal is only to output one bit $h_1 \in \{0,1\}$.

Privately learning parities was first studied by \citet{kasiviswanathan2011can}.
We will use this problem to asymptotically separate our algorithm from sample-and-aggregate.

We cannot learn $h$ with fewer than $d$ samples -- this is a basic fact of linear algebra.
More generally, given $d-k$ samples, the probability that we can determine $h_1$ is at most $2^{-k}$.
If we cannot determine $h_1$, we can guess it randomly, so the overall accuracy is upper bounded by $\frac12 + 2^{-k-1}$.
Thus sample-and-aggregate requires $n=\Omega(d/\varepsilon)$ samples to output $h_1$ with high accuracy, since each subsample needs size $\Omega(d)$.

On the other hand, with $d+\lceil\log_2(1/\beta)\rceil$ samples we can learn $h$ with probability at least $1-\beta$; this is our non-private function that we insert into Theorem \ref{thm:main}.
Setting $\beta=1/O(k)$ we get a sample complexity of $n=m+d+O(\log k)$ with $m \ge t = O\big( \frac1\varepsilon \log(1/\delta) \big)$ (and $m \le n$) arbitrary (since $\mathcal{Y}=\{0,1\}$, the $\log^*$ disappears) and using $k = C(n,m,t) \le n^t$ evaluations of the function. This simplifies to $n = m + d + O(\frac1\varepsilon \log(1/\delta) \log (n))$.
Thus our algorithm beats sample-and-aggregate when $m+\frac1\varepsilon \log(1/\delta) \log (n) \le o(d/\varepsilon)$.
Asymptotically,  take $m=d \to \infty$ and $\varepsilon=1/\log d$ and $\delta = 1/d^{10}$ to get a separation of $n\ge \Omega(d \log d)$ for sample-and-aggregate vs $n=O(d)$ for our algorithm, with $k=C(O(d),d,O(\log^2d)) \le d^{O(\log d)}$ queries.

\section{Lower Bound}

Now we prove our lower bound which shows the near-optimality of our upper bound.
Figure \ref{fig:upperlowerbounds} shows how this bound looks numerically.

\begin{figure}
    \centering
    \includegraphics[width=0.5\linewidth]{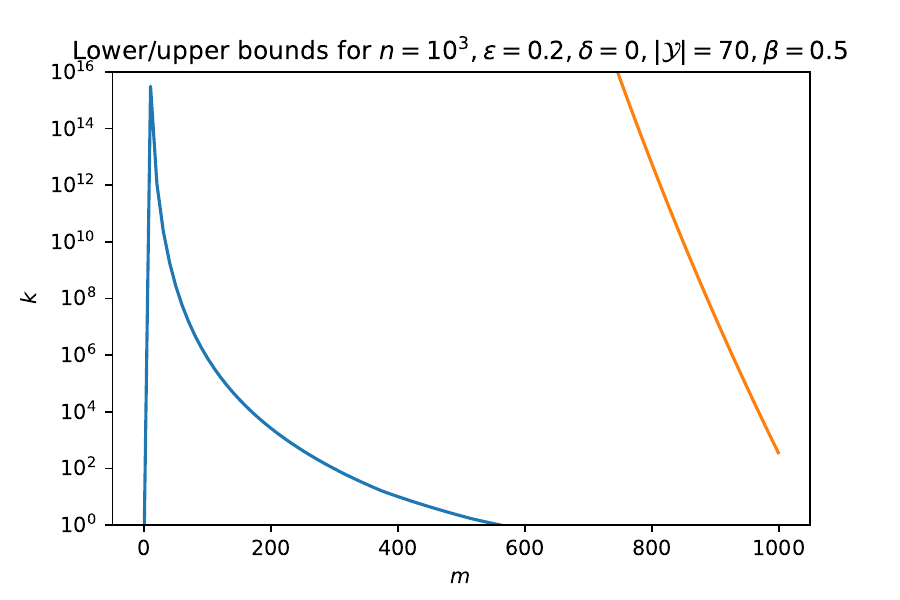}%
    \includegraphics[width=0.5\linewidth]{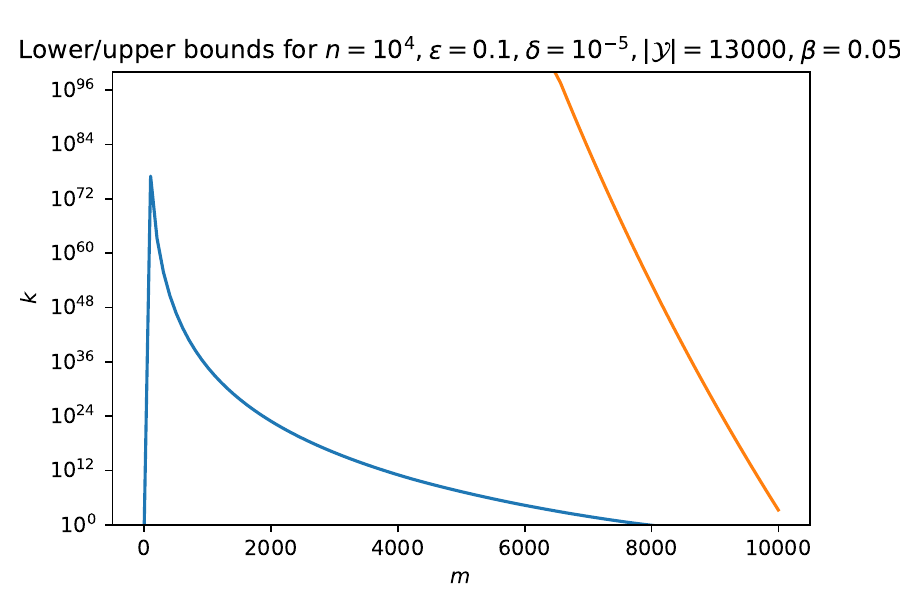}
    \caption{
    Comparing the query complexity lower bound from Theorem \ref{thm:lower-general} (taking a max over $t$) with the upper bound from Theorem \ref{thm:main-formal-pdp} (which satisfies pure DP).
    The parameters are chosen according to Theorem \ref{thm:main-formal-pdp} so that on the left $t = 2\left\lceil \frac{2}{\varepsilon} \log\left(\frac{|\mathcal{Y}|}{\beta}\right) \right\rceil = 100$ and on the right $t=500$, which allows comparison with Figure \ref{fig:max-bound}.
    %Note that the upper and lower bounds are not perfectly comparable (see discussion on page \pageref{discuss-range-lowerbound}) because the upper bound assumes bounds on $|\mathcal{Y}|$ that the lower bound does not; this explains why the lines are so close on the left.
    }
    \label{fig:upperlowerbounds}
\end{figure}

\begin{theorem}[Lower Bound -- General Version] \label{thm:lower-general}
    Let $\mathcal{Y} \subset \mathbb{Z}$ be finite with at least $2$ elements.
    Let $M^f : \mathbb{Z}^n \to \mathcal{Y}$ be a randomised algorithm that makes at most $k$ queries to an oracle $f : \mathbb{Z}^* \to \mathcal{Y}$. Let $\gamma>0$ be fixed.
    Suppose that, for every oracle $f$, the algorithm $M^f$ satisfies $(\varepsilon,\delta)$-differential privacy and the following.
    Let $\mathcal{D}$ be an arbitrary distribution on $\mathbb{Z}$ and let $\nu \in \mathcal{Y}$. If $\pr{X \gets \mathcal{D}^{n-m}}{f(X)=\nu} \ge 1-\gamma$, then $\pr{X \gets \mathcal{D}^n}{M^f(X)=\nu} \ge 1/2$.
    Then we must have \[k ~ \ge ~ \max_{t \in [m]} ~~ \frac{{n \choose t}}{{m \choose t}} \cdot \left(\frac12 e^{-2t\varepsilon} - \frac{\delta}{e^{\varepsilon}-1}  - \frac{1}{|\mathcal{Y}|-1} \right).\] %for all integers $t \in [m]$.
\end{theorem}

Theorem \ref{thm:lower} in the introduction is attained by substituting either $t = 1/\varepsilon $ or $t= \log(1/\delta)/200\varepsilon $ into Theorem \ref{thm:lower-general} and taking $\mathcal{Y}$ to be a sufficiently large set of integers.

We begin by giving some intuition for the lower bound:
Consider $\mathcal{Y}=\{0,1\}$ and the function $f : \{0,1\}^* \to \{0,1\}$ being the $\max$ function and consider the data distribution $\mathcal{D}$ to be a point mass on $0$.
Clearly, $f$ always evaluates to $0$ on input from this distribution and thus our differentially private algorithm, given $n$ samples from this distribution, should output $0$ with high probability.
Assume, for now, that $M$ is restricted to evaluating $f$ on subsets of its input of the appropriate size.
Now suppose $t \approx \frac{1}{\varepsilon}\log(1/\delta)$ out of the $n$ samples are corrupted to be $1$ instead of 0.
If each subset that $M$ evaluates $f$ on includes at least one corrupted sample, then each evaluation will return $1$ and so $M$ should output $1$ with high probability.
This sets up a contradiction with group privacy, since we have two inputs that only differ by $t$ replacements on which the outputs are very different. 
The only way to avoid the contradiction is for the algorithm to query enough sets so that at least one of them doesn't contain any corrupted samples -- in a sense, the subsets queried must form a covering design -- and that's the lower bound.

The above proof sketch has a couple of holes that we must fill:
(i) We must vary the correct answer to rule out an algorithm that simply ``knows'' that the correct answer is $0$. 
(ii) We must randomise which inputs are corrupted to rule out an algorithm that ``knows'' which inputs to avoid.
(iii) The algorithm could defeat this particular setup by evaluating $f$ on small subsets, which are more likely to exclude all the corrupted samples, so we need to modify the function $f$ to ``fail'' when given the wrong input size.
(iv) The algorithm could also defeat this particular setup by evaluating the function on ``fake'' inputs, so we must rule this out too.
(v) Finally, we need to be careful with the group privacy argument; we are effectively performing a ``packing argument'' \cite{hardt2010geometry}.

\begin{proof}
    Let $\ell$ be a sufficiently large integer.\footnote{We will eventually take $\ell \to \infty$.}
    Let $S \subseteq [\ell^2]$ be  uniformly random with size $|S|=\ell$.
    Define $\mathcal{D}_S$ to be the uniform distribution on $S$.
    Without loss of generality, assume $0 \in \mathcal{Y}$.
    Fix $\nu \in \mathcal{Y} \setminus \{0\}$ to be chosen later.
    Let $f_{S,\nu,n-m} : \mathbb{Z}^* \to \mathcal{Y}$ be defined as
    \[
        f_{S,\nu,n-m}(x) = \left\{ \begin{array}{cl} \nu & \text{ if } |x| = n-m \text{ and } \forall i ~ x_i \in S \text{ and } \forall i \ne j ~ x_i \ne x_j\\ 0 & \text{ if } |x| \ne n-m \text{ or } \exists i ~ x_i \notin S \text{ or } \exists i \ne j ~ x_i = x_j \end{array} \right\}.
    \]
    In words, $f_{S,\nu,n-m}(x) = \nu$ when the size of the input $x$ is exactly $n-m$ (without repetitions) and all of the elements in $x$ are inputs from $S$; otherwise $f_{S,\nu,n-m}(x) = 0$.
    
    By construction, $f_{S,\nu,n-m}(X) = \nu$ whenever $X \gets \mathcal{D}_S^{n-m}$ and there are no collisions, which happens with probability \[\!\!\pr{X \gets \mathcal{D}_S^{n-m}}{f_{S,\nu,n-m}(X) = \nu} = \!\!\! \pr{X \gets \mathcal{D}_S^{n-m}}{\forall i \ne j ~ X_i \ne X_j} = \!\!\! \prod_{i=0}^{n-m-1} \frac{|S|-i}{|S|} \ge 1 - \frac{(n-m)^2}{2\ell} \ge 1-\gamma,\] assuming $\ell \ge n^2/2\gamma$. Hence, by assumption, $M^{f_{S,\nu,n-m}}(X) = \nu$ with probability at least $1/2$ whenever $X \gets \mathcal{D}_S^n$. %(In other words, $M^{f_{S,\nu,n-m}}(X)$ rounds to $\nu$ a majority of the time when $X \gets \mathcal{D}_S^n$.)
    
    Now we define a new ``corrupted'' distribution $\mathcal{C}_{S,n,t}$ on $\mathcal{X}^n$ as follows. It consists of $n-t$ independent uniformly random samples from $S$ and $t$ independent uniformly random samples from $[\ell^2]\setminus S$ in an independently uniformly random order.
    By construction, there is a coupling between $\mathcal{C}_{S,n,t}$ and $\mathcal{D}_S^n$ such that they always differ by $t$ replacements or $2t$ additions/removals -- i.e., the $\infty$-Wasserstein distance between the distributions is bounded by $2t$.
    Given $S$, it is easy to distinguish $\mathcal{C}_{S,n,t}$ from $\mathcal{D}_S^n$. However, if $S$ is unknown, these distributions are indistinguishable (provided $\ell$ is sufficiently large). 
    Our key claim is that, given a sample $\widetilde{X} \gets \mathcal{C}_{S,n,t}$ and oracle access to $f_{S,\nu,n-m}$ (and no additional information about $S$), it is hard to generate a query to the oracle that returns the value $\nu$:

    \begin{claim}\label{claim:corrupt}
        Let $g : \mathbb{Z}^* \to \mathbb{Z}^*$ be a (possibly randomised) function. 
        Let $S \subseteq [\ell^2]$ be uniformly random of size $|S|=\ell$.
        Let $\widetilde{X} \gets \mathcal{C}_{S,n,t}$ -- that is, $\widetilde{X}$ contains $n-t$ elements from $S$ and $t$ elements from $[\ell^2]\setminus S$ and is otherwise uniformly random.
        Let $Y = g(\widetilde{X})$. Then \[\pr{}{|Y|=n-m \text{ and } \forall i ~ Y_i \in S \text{ and } \forall i \ne j ~ Y_i \ne Y_j} \le \frac{{m \choose t}}{{n \choose t}} + \frac{n^2}{2\ell},\] assuming $\ell$ is sufficiently large and $0 \le t \le m \le n$.
    \end{claim}
    \begin{proof}
        For simplicity, we split the analysis into two cases depending on whether or not $\widetilde{X}$ has any collisions. Let $E$ denote the event that there are collisions and let $\overline{E}$ denote the event that there are no collisions. The probability of no collisions is 
        \begin{align}
            \pr{\widetilde{X}\gets\mathcal{C}_{S,n,t}}{\overline{E}} 
            &=
            \pr{\widetilde{X}\gets\mathcal{C}_{S,n,t}}{\forall i \ne j ~ \widetilde{X}_i \ne \widetilde{X}_j} \\
            &= \left(\prod_{i=0}^{n-t-1} \frac{|S|-i}{|S|}\right)\left(\prod_{j=0}^{t-1} \frac{|[\ell^2]\setminus S|-j}{|[\ell^2]\setminus S|}\right) \notag\\
            &\ge 1 - \sum_{i=0}^{n-t-1} \frac{i}{\ell} - \sum_{j=0}^{t-1} \frac{j}{\ell^2-\ell} \tag{Bernoulli's inequality} \\
            &\ge 1 - \frac{(n-t)^2}{2\ell} - \frac{t^2}{2\ell(\ell-1)} \ge 1 - \frac{n^2}{2\ell}. \label{eq:probnocoll}
        \end{align}
        Since the probability of a collision $\pr{}{E}\le n^2/2\ell$ is low (because $\ell$ is large), we focus on the case of no collisions.

        Now we take a Bayesian perspective. Conditioned on $\widetilde{X}$, the set $S$ and the output $Y$ are independent. We want to upper bound the probability of the event that $Y$ consists of $n-m$ distinct points all in $S$. We can consider a fixed $Y$; and we can assume $|Y|=n-m$ and $\forall i \ne j ~ Y_i \ne Y_j$ (since otherwise the probability of the event is zero).
        Given $\widetilde{X}$ (with no collisions -- i.e., $|\widetilde{X}|=n$), all we know about $S$ is that its intersection with $\widetilde{X}$ has size exactly $|S \cap \widetilde{X}|=n-t$; otherwise $S$ is uniformly random with size $|S|=\ell$.
        Now we can directly calculate the probability that $\forall i ~Y_i \in S$ based on the intersection between $Y$ and $\widetilde{X}$, namely
        \begin{align}
            \prc{S}{\forall i ~ Y_i \in S}{\overline{E}} &= \prc{S}{Y \subseteq S}{\overline{E}} \notag\\
            &= \prc{S \cap \widetilde{X}}{Y \cap \widetilde{X} \subseteq S \cap \widetilde{X}}{\overline{E}} \prc{S \setminus \widetilde{X}}{Y \setminus \widetilde{X} \subseteq S \setminus \widetilde{X}}{\overline{E}} \notag\\
            &= \frac{{n-|Y \cap \widetilde{X}| \choose n-t-|Y \cap \widetilde{X}|}}{{n \choose n-t}} \frac{{\ell^2-n - |Y \setminus \widetilde{X}| \choose \ell-(n-t) - |Y \setminus \widetilde{X}|}}{{\ell^2-n \choose \ell-(n-t)}} \notag\\
            &= \frac{{n-|Y \cap \widetilde{X}| \choose t}}{{n \choose t}} \frac{{\ell^2-n - |Y| + | Y \cap \widetilde{X}| \choose \ell^2-\ell-t}}{{\ell^2-n \choose \ell^2-\ell-t}} =: h(|Y \cap \widetilde{X}|).\label{eq:probcontain}
        \end{align}
        Next we show that the function $h$ defined above is increasing:
        For $k < n-t \le n$, we have
        \begin{align}
            \frac{h(k+1)}{h(k)} &= \frac{{n-(k+1) \choose t} {\ell^2 - \ell - |Y| + (k+1) \choose \ell^2-n-t}}{{n-k \choose t} {\ell^2 - n - |Y| + k \choose \ell^2-\ell-t}} \notag\\
            &= \frac{(n-k-t)(\ell^2-n-|Y|+k+1)}{(n-k)(\ell+t-n+k+1-|Y|)} \notag\\
            &\ge 1,
        \end{align}
        assuming $\ell$ is sufficiently large.
        Thus Equation \ref{eq:probcontain} is maximised when $k=|Y \cap \widetilde{X}| = |Y| = n-m$. Hence
        \[\prc{S}{\forall i ~ Y_i \in S}{\overline{E}} = \prc{S}{Y \subseteq S}{\overline{E}} \le \frac{{n-(n-m) \choose t}}{{n \choose t}} \frac{{\ell^2-n \choose \ell^2-\ell-t}}{{\ell^2-n \choose \ell^2-\ell-t}} = \frac{{m \choose t}}{{n \choose t}}.\]
        The above analysis assumes no collisions in $\widetilde{X}$. 
        Rather than carefully analyzing the case with collisions, we use the naive bound $\prc{S}{\forall i ~ Y_i \in S}{E} \le 1$ for this case. Combining with the bound of Equation \ref{eq:probnocoll} gives \[\pr{S}{\forall i ~ Y_i \in S} \le  \prc{S}{\forall i ~ Y_i \in S}{\overline{E}} \cdot \pr{}{\overline{E}} + \prc{S}{\forall i ~ Y_i \in S}{E} \cdot \pr{}{E} \le \frac{{m \choose t}}{{n \choose t}} \cdot 1 + 1 \cdot \frac{n^2}{2\ell},\] as required.
    \end{proof}
    It follows from Claim \ref{claim:corrupt} and a union bound that, if $M^{f_{S,\nu,n-m}}$ makes at most $k$ queries to the oracle $f_{S,\nu,n-m}$, then the probability that the oracle \emph{ever} returns something other than $0$ is at most $k \cdot \left(\frac{{m \choose t}}{{n \choose t}} + \frac{n^2}{2\ell}\right)$. Hence we can \emph{almost} simulate $M^{f_{S,\nu,n-m}}$ by running $M$ with an oracle that always returns $0$.
    Namely, for all $V \subseteq \mathcal{Y}$, \[\pr{S, \widetilde{X} \gets \mathcal{C}_{S,n,t}}{M^{f_{S,\nu,n-m}}(\widetilde{X}) \in V} \le \pr{S, \widetilde{X} \gets \mathcal{C}_{S,n,t}}{M^{0}(\widetilde{X}) \in V} + k \cdot \left(\frac{{m \choose t}}{{n \choose t}} + \frac{n^2}{2\ell}\right).\label{eq:nu-eta}\]
    Note that the right hand side of Equation \ref{eq:nu-eta} does \emph{not} depend on $\nu$. 
    Now we pick an arbitrary $\nu \in \mathcal{Y} \setminus \{0\}$ such that \[\pr{S, \widetilde{X} \gets \mathcal{C}_{S,n,t}}{M^{0}(\widetilde{X}) = \nu} \le \frac{1}{|\mathcal{Y} \setminus \{0\}|}.\]
    
    By group privacy (Lemma \ref{lem:group-privacy}), for all $V \subseteq \mathbb{Z}$ and all $S \subseteq [\ell^2]$ with $|S|=\ell$, we have
    \[\pr{X \gets \mathcal{D}_S^n}{M^{f_{S,\nu,n-m}}({X}) \in V} \le e^{2t\varepsilon} \pr{\widetilde{X} \gets \mathcal{C}_{S,n,t}}{M^{f_{S,\nu,n-m}}(\widetilde{X}) \in V} + \frac{e^{2t\varepsilon}-1}{e^\varepsilon-1} \delta.\]
    (This relies on the fact that we can couple $X \gets \mathcal{D}_S$ and $\widetilde{X} \gets \mathcal{C}_{S,n,t}$ such that \\$\pr{}{|X \setminus \widetilde{X}| = |\widetilde{X} \setminus X| = t} = 1$.)
    By our accuracy assumption, \[\pr{X \gets \mathcal{D}_S^n}{M^{f_{S,\nu,n-m}}({X}) = \nu } \ge \frac12.\] 
    Now we string together the above equations to obtain
    \[\frac12 \le e^{2t\varepsilon} \left( \frac{1}{|\mathcal{Y} \setminus \{0\}|} + k \cdot \left(\frac{{m \choose t}}{{n \choose t}} + \frac{n^2}{2\ell}\right) \right) + \frac{e^{2t\varepsilon}-1}{e^\varepsilon-1} \delta.\]
    At this point we can take the limit $\ell \to \infty$ to obtain
    \[\frac12 \le \frac{e^{2t\varepsilon}}{|\mathcal{Y} \setminus \{0\}|} + e^{2t\varepsilon} \cdot k \cdot \frac{{m \choose t}}{{n \choose t}} + \frac{e^{2t\varepsilon}-1}{e^\varepsilon-1} \delta \le e^{2t\varepsilon} \cdot \left( \frac{1}{|\mathcal{Y}|-1} + k \cdot \frac{{m \choose t}}{{n \choose t}} + \frac{\delta}{e^\varepsilon-1} \right),\]
    which rearranges to \[k \ge \frac{{n \choose t}}{{m \choose t}} \cdot \left(\frac12 e^{-2t\varepsilon} - \frac{\delta}{e^\varepsilon-1} - \frac{1}{|\mathcal{Y}|-1} \right),\] as required.
\end{proof}

%Our lower bound in Theorem \ref{thm:lower-general} assumes the algorithm $M$ has an infinite range $\mathcal{Y}=\mathbb{Z}$. \label{discuss-range-lowerbound}
%In contrast, our algorithm (Theorem \ref{thm:main-formal-adp}) assumes a finite range $\mathcal{Y}$ and has a dependency on its size, namely $t=\frac1\varepsilon \log(1/\delta) \exp(O(\log^* |\mathcal{Y}|))$.
%However, the proof of Theorem \ref{thm:lower-general} assumes the algorithm $M$ has range $\mathcal{Y}=[\ell^2]$ (and we take $\ell \to \infty$). Thus our lower bound could be extended to the setting with a finite range.

The main gap between our upper and lower bounds is the dependence on the size of the range $\mathcal{Y}$.
In our upper bound, we have $t=\frac1\varepsilon \log(1/\delta) \exp(O(\log^*|\mathcal{Y}|))$.
In our lower bound, we have no such dependence. (The size of the range $\mathcal{Y}$ does appear in Theorem \ref{thm:lower-general}. However, its influence is limited. Namely, if we take $|\mathcal{Y}| \to \infty$ we do \emph{not} get $k \to \infty$ in the lower bound.) 

That said, some dependence on the size of the range $\mathcal{Y}$ is necessary:  \citet{bun2015differentially,alon2019private} show that the interior point problem -- basically, computing an approximate median to constant accuracy -- requires $n=\Omega(\log^*|\mathcal{Y}|)$ samples under $(\varepsilon=0.1,\delta=0.01/n^2)$-differential privacy. 
Our algorithm can be used to compute an approximate median under differential privacy.
Ergo the $n=\Omega(\log^*|\mathcal{Y}|)$ lower bound applies to our algorithm.
However, this only gives us a lower bound on the number of data points $n$ and does not seem to imply anything nontrivial about the number of queries $k$.
%Thus, while the $\exp(O(\log^* |\mathcal{Y}|))$ term in the upper bound could potentially be improved, it cannot be removed entirely.

\section{Discussion} \label{sec:discussion}

To recap: Our main result (Theorem \ref{thm:main}) provides a differentially private algorithm which takes a real-valued black-box function $f$ and a private dataset $x$ (consisting of i.i.d.~samples from some distribution $\mathcal{D}$) and evaluates the function on $k$ subsets of the input dataset and then outputs a statistical estimate $y \approx f(\mathcal{D}^{n-m})$ for the value of the function. 
Our result trades off between oracle efficiency (i.e., how many subsets we evaluate the function on) and statistical efficiency (i.e., the size of each of those subsets).
We also prove a lower bound (Theorem \ref{thm:lower}) that shows that our upper bound is roughly optimal. Namely, the combinatorial term appearing in the oracle complexity $k \approx \frac{{n \choose t}}{{m \choose t}}$ is necessary, where $t$ depends on the differential privacy parameters.

\subsection{Interpretation}
By varying the parameter $m$ (with $n \ge m \ge t \approx \Theta(\frac1\varepsilon\log(1/\delta))$) our result interpolates between sample-and-aggregate \cite{NRS07} and more recent results \cite{FangDY22,LRSS25}.
While we believe that the entire tradeoff curve is interesting, we point out a few illustrative values along the curve in Table \ref{tab:tradeoff} to aid understanding. We also plot some numerical values of the upper and lower bounds in Figure \ref{fig:upperlowerbounds}.

Figure \ref{fig:upperlowerbounds} shows that the number of evaluations $k=C(n,m,t)$ can become extremely large. In particular, the practical parameter regime is where $m$ is close to $n$.

\begin{table}[h]
    \centering
    \begin{tabular}{|c|c|c|c|}
        \hline
         &Subset size & Number of evaluations &  \\
         &$n-m$ & $k=C(n,m,t)$ & Note\\
         &(larger is better) & (smaller is better) & \\
         \hline \hline
         1&$\frac{n}{t+1}$ & $t+1$ & Cf. sample-and-aggregate \cite{NRS07} \\
         \hline
         2&$\frac{2n}{t+2}$ & $\le \frac{(t+2)(t+1)}{2} = O(t^2)$ & (\ref{eq:m=t},\ref{eq:chunks})\\
         \hline         
         3&$\frac{cn}{t+c}$ & $\le O(\min\{ e^{t+c} \cdot t \log m, t^c \} )$ & $c \in \mathbb{N}$ \\
         \hline
         4&$n-ct$ & $\le O\left(\left(\frac{n}{ct}\right)^t \cdot ct\right)$ & $c \in \mathbb{N}$ \\
         \hline
         5&$n-t$ & ${n \choose t} = O(n/t)^t$ & \cite[Cf.][]{LRSS25} \\
         \hline
         6& & $< \frac{{n \choose t}}{{m \choose t}} \left(1 + \log{ m \choose t} \right) +1$ & general upper bound ($n \ge m \ge t$) \\
         \hline
    \end{tabular}
    \caption{Example parameter choices for Theorem \ref{thm:main}. This shows the tradeoff between the number of evaluations of the function $k$ (i.e., oracle complexity) and the size of the subsets on which we evaluate it (which determines statistical efficiency). Here $t=\frac{1}{\varepsilon} \log(1/\delta) \exp(O(\log^* |\mathcal{Y}|))$ depends on the privacy parameters $\varepsilon, \delta>0$ and the size of the range $\mathcal{Y}$ of the function.}
    \label{tab:tradeoff}
\end{table}

Perhaps the most practically relevant instantiations of our result are given in Lines 2 and 3 of Table \ref{tab:tradeoff}. Compared to Line 1, these show that we can increase the subset size by a constant factor while only suffering a polynomial blowup in the number of evaluations. 
In particular, we can roughly double the amount of data available in each evaluation at the expense of only a quadratic blowup in the number of evaluations.

On the other end of the tradeoff curve, we can compare Lines 4 and 5 of Table \ref{tab:tradeoff}. The amount of data that is ``sacrificed'' for differential privacy increases by a factor of $c$. This decreases the number of evaluations by a multiplicative factor of $c^t$, which is significant, but the large $n^t$ factor remains.

\subsection{Limitations \& Further Work}\label{sec:limitations}
As mentioned in the introduction (page \pageref{nb:computational-efficiency}, bullet point 5), the main limitation of our algorithm is that, while we bound the oracle complexity (i.e., the number of evaluations of the function), we do not account for the computational cost of choosing the subsets of the input on which to evaluate the function and of processing the values returned by the function.
Note that in many cases evaluating the function can be quite expensive -- e.g., in the PATE framework \cite{pate1,pate2}, each function evaluation corresponds to training a machine learning model -- thus it is reasonable to focus on minimizing the number of evaluations.

Choosing the subsets amounts to generating a covering design.
As mentioned in Section \ref{sec:designs}, we do not have general-purpose optimal existential results for covering designs, let alone efficient constructions. (Although there are a lot of special-purpose constructions in the literature.)
However, a (slightly suboptimal) covering design can be constructed (with high probability) by simply taking enough random subsets of the appropriate size (see the proof of Proposition \ref{prop:covering-size} and Footnote \ref{fn:cover-prob}). We must account for the failure probability, but otherwise this does give an efficient method for choosing the subsets.\footnote{Verifying that a given collection of sets is a covering design is co-NP-hard. Thus this failure probability cannot easily be ``checked away.'' However, we remark that the privacy guarantee of our algorithm still holds if the sets do not form a covering design; namely Lemma \ref{lem:mon} holds for any choice of subsets $S_1, \cdots, S_k \subseteq [n]$.}

Next we consider the computational complexity of processing the function values.
Recall Algorithm \ref{alg:estimate}: 
We have a $(n,m,t)$-covering design $S_1, \cdots, S_k \subseteq [n]$ and we evaluate $f(x_{[n]\setminus S_i})$ for each $i \in [k]$.
From these values, we must compute \[g(x') := \max\{ f(x'_{[n] \setminus S_i}) : i \in [k] , |x'_{[n]\setminus S_i}|=n-m \}\] where $x' \subseteq x$.
Then the shifted inverse mechanism (see Appendix \ref{app:shi} for details) computes \[\ell(x,y) := \min \{ |x \setminus \widetilde{x}| : \widetilde{x} \subseteq x, g(\widetilde{x}) \le y\}.\]
This task can be framed as a decision problem: 
\begin{align}
    \!\!\ell(x,y) \le v \!&\iff\! \exists \widetilde{x} \subseteq x ~~ |x \setminus \widetilde{x} | \le v \wedge g(\widetilde{x}) \le y \notag\\
    &\iff\! \exists T \subseteq [n] ~~ | T | \le v \wedge g(x_{[n] \setminus T}) \le y \notag\\
    &\iff\! \exists T \subseteq [n] ~~ | T | \le v \wedge \big( \forall i \in [k] ~ |(x_{[n] \setminus T})_{[n] \setminus S_i}| = n-|S_i| \implies f(x_{[n] \setminus S_i}) \le y  \big) \notag\\
    &\iff\! \exists T \subseteq [n] ~~ | T | \le v \wedge \big( \forall i \in [k] ~ T \subseteq S_i \implies f(x_{[n] \setminus S_i}) \le y  \big) \notag\\
    &\iff\! \exists T \subseteq [n] ~~ | T | \le v \wedge \big( \forall i \in [k] ~ f(x_{[n] \setminus S_i}) > y \implies T \not\subseteq S_i \big) \notag\\
    &\iff\! \exists T \subseteq [n] ~~ | T | \le v \wedge \big( \forall i \in [k] ~ f(x_{[n] \setminus S_i}) > y \implies T \cap ([n] \setminus S_i) \ne \emptyset \big) . \label{eq:decision}
\end{align}
In words, proving $\ell(x,y) \le v$ is equivalent to finding a set $T$ of at most $v$ points such that at least one point lies in each set of the form $[n] \setminus S_i$ with $f(x_{[n] \setminus S_i})>y$ and $i \in [k]$.
That is, the set $T$ is a ``hitting set'' for all the input sets corresponding to function values that are larger than $y$. In general, finding a small hitting set is equivalent to the set cover problem, which is NP-complete \cite{Karp1972}.

The fact that processing the function values reduces to an NP-complete problem suggests that our algorithm cannot be made computationally efficient.
However, this reduction is going the wrong way to make that suggestion formal -- that is, it does \emph{not} prove NP-hardness. Hope is not lost!

If the collection of sets $S_1, \cdots, S_k$ is arbitrary and the function $f$ and input $x$ are arbitrary, then Equation \ref{eq:decision} can give rise to arbitrary instances of the hitting set problem, which is NP-complete. Thus, to avoid NP-hardness, we need to rely on some of these parameters not being arbitrary.

In the black-box setting the function $f$ is indeed arbitrary. 
The input $x$ is i.i.d., but from an arbitrary distribution. 
Crucially, the sets $S_1, \cdots, S_k$ are \emph{not} arbitrary; our algorithm can choose them, subject only to the constraint that they form a covering design.
This points to an avenue for making the algorithm computationally efficient -- construct a covering design with additional structural properties that make the decision problem in Equation \ref{eq:decision} easy. We formulate this precisely as an open problem:

\begin{openproblem}
    Construct a pair of algorithms $\mathsf{Gen}$ and $\mathsf{Eval}$ with the following properties.
    \begin{itemize}
        \item Given integers $n \ge m \ge t$, $\mathsf{Gen}$ produces a $(n,m,t)$-covering design (Definition \ref{def:covering_design}) $S_1, \cdots, S_k$ (and also outputs some $\mathsf{context}$ to pass to $\mathsf{Eval}$). That is, \[(S_1, \cdots, S_k, \mathsf{context}) \gets \mathsf{Gen}(n,m,t),\label{eq:op1}\] where $S_1, \cdots, S_k \subseteq [n]$ and $\forall i \in [k] ~ |S_i|=m$ and $\forall T \subseteq [n] ~ |T|\le t \implies \exists i \in [k] ~ T \subseteq S_i$.
        \item The number of sets $k$ should be not too large; ideally $k \le C(n,m,t) \cdot \mathrm{poly}(n)$, where $C(n,m,t)$ is the smallest possible $k$.
        \item Given $v \in [n]$ and $I \subseteq [k]$, $\mathsf{Eval}$ indicates whether the collection $[n] \setminus S_i$ for $i \in I$ has a hitting set of size $\le v$. That is, \[\mathsf{Eval}(v, I, \mathsf{context}) = \mathsf{true} \iff \exists T \subseteq [n] ~ |T| \le v \wedge \forall i \in I ~ T \cap ([n] \setminus S_i) \ne \emptyset.\label{eq:op2}\]
        \item Both $\mathsf{Gen}$ and $\mathsf{Eval}$ should be computationally efficient for the parameter regime of interest. Ideally, their runtime should be polynomial in the parameter $n$ and the covering size $k$. (Note that $k \ge C(n,m,t)$ may be exponential in $n$ per Proposition \ref{prop:covering-size}.)
    \end{itemize}
    The algorithms  $\mathsf{Gen}$ and $\mathsf{Eval}$ may be randomised; it suffices for the guarantees in Equations \ref{eq:op1} and \ref{eq:op2} to hold with high probability, although the dependence of $k$ and the algorithms' runtimes on the failure probability should be polylogarithmic.
\end{openproblem}

Tying the open problem back to our application:
Algorithm \ref{alg:estimate} would first call $\mathsf{Gen}$ to produce the sets $S_1, \cdots, S_k$ (and $\mathsf{context}$). Then it would evaluate $f(x_{[n] \setminus S_i})$ for each $i \in [k]$. Finally, it would run the shifted inverse mechanism (see Appendix \ref{app:shi} for details); this requires evaluating $\ell(x,y)$ for various values of $y$, which depends on $f$. Per Equations \ref{eq:decision} and \ref{eq:op2}, \[\ell(x,y) \le v \iff \mathsf{Eval}(v,I,\mathsf{context})=\mathsf{true}, ~\text{ where }~ I = \{ i \in [k] : f(x_{[n] \setminus S_i}) > y \}.\]
Thus $\ell(x,y)$ can be evaluated by calling $\mathsf{Eval}$ and performing binary search on $v \in \{0,1,2,\cdots,n\}$. Overall, to implement Theorem \ref{thm:main-formal-cdp} the runtime of Algorithm \ref{alg:estimate} would be dominated by the runtime of one call to $\mathsf{Gen}$, plus $k$ calls to $f$, plus $O(\log(n) \cdot \log|\mathcal{Y}|)$ calls to $\mathsf{Eval}$.\footnote{The concentrated differentially private version of our algorithm in Theorem \ref{thm:main-formal-cdp} performs binary search on $y \in \mathcal{Y}$ (see Appendix \ref{app:noisy-binary-search} for details) and so we only evaluate $\ell(x,y)$ for $O(\log|\mathcal{Y}|)$ values of $y$. The other versions of our algorithm potentially require evaluating $\ell(x,y)$ for more values of $y$.}

Finally, we remark that, although set cover is NP-complete, the runtime is exponential only in the size of the covering, which should be modest (at most $t$). That is, a na\"ive implementation of $\mathsf{Eval}$ would run in time $O(n^v \cdot k)$.

\newpage

\addcontentsline{toc}{section}{References}
\begin{small}
\printbibliography    
\end{small}

\newpage

\appendix

\section{Shifted Inverse Mechanism}\label{app:shi}

In Section \ref{sec:shi}, we stated the properties of the shifted inverse mechanism of \citet{FangDY22}, which is integral to our algorithm.
We now briefly review how this algorithm works, following the presentation of \citet{DPorg-down-sensitivity}.

The key idea behind the shifted inverse mechanism is the following transformation from an arbitrary monotone function to a low-sensitivity loss function. The benefit of this transformation is that low-sensitivity functions are something we know how to work with in a differentially private manner.

\begin{proposition}[{\cite[Lemma 4.1]{FangDY22},\cite[Proposition 3]{DPorg-down-sensitivity},\cite[Lemma 3.4]{LRSS25}}]\label{prop:shi-loss}
    Let $g : \mathcal{X}^* \to \mathbb{R}$ be monotone -- i.e., $x' \subseteq x \implies g(x') \le g(x)$.
    Define $\ell : \mathcal{X}^* \times \mathbb{R} \to \mathbb{Z} \cup \{\infty\}$ by
    \[\ell(x,y) := \min \{ |x \setminus \widetilde{x}| : \widetilde{x} \subseteq x, g(\widetilde{x}) \le y\}.\label{eq:shi}\]
    Then $\ell$ has sensitivity $1$ in its first argument. That is, $|\ell(x,y)-\ell(x',y)|\le |x \setminus x'|+|x' \setminus x|$ for all $x,x' \in \mathcal{X}^n$ and all $y \in \mathbb{R}$ with $y \ge g(\emptyset)$.
\end{proposition}

The transformation \eqref{eq:shi} is invertible, namely \[ g(x) = \min \{ y \in \mathbb{R} : \ell(x,y) = 0 \}\label{eq:shi-inversion}\] for all $x \in \mathcal{X}^*$. Note that $\ell(x,y) \ge 0$ is a decreasing function of $y$.

Essentially, the shifted inverse mechanism works by performing this inversion \eqref{eq:shi-inversion}. Of course, under the constraint of differential privacy, we can only approximate $\ell(x,y)$ and hence we can only approximately perform the inversion. Performing this inversion is roughly equivalent to computing a differentially private approximate median, which is a well-studied task.

To be precise, the shifted inverse mechanism finds $y$ satisfying two conditions:
\begin{itemize}
    \item First, $\ell(x,y) \le t$ for some tolerance $t>0$, which is equivalent to $y \ge \min \{g(x') : x' \subseteq x, |x'| \ge |x|-t\}$.
    \item Informally, we want $\ell(x,y)>0$, which is equivalent to $y < g(x)$. 
    Formally, define \[\overline{\ell}(x,y) := \min \{ \ell(x,y-\eta) :\eta>0 \} = \min \{ |x \setminus \widetilde{x}| : \widetilde{x} \subseteq x, g(\widetilde{x}) < y\}.\label{eq:shi-loss2}\]
    Then our second condition is $\overline{\ell}(x,y)>0$, which is equivalent to $y \le g(x)$. Note that $\overline{\ell}$ has the same properties as $\ell$, namely it has sensitivity $1$ in its first argument and is decreasing in its second argument.
\end{itemize}
Note that the shifted inverse mechanism \emph{underestimates} $g(x)$.
Some bias in differentially private estimation is inherent \cite{kamath2025bias}; the negative direction of the bias arises from the fact that we are looking at down-local algorithms (i.e., we are removing some input elements) and from the monotonicity of $g$.

There are several differentially private ways to approximately implement the inversion \eqref{eq:shi-inversion}, which we review next. These lead to the various forms of the shifted inverse mechanism in Section \ref{sec:shi}.

Note that, for the differentially private inversion to work, we must restrict to a finite search space \citep{bun2015differentially,alon2019private}. Thus we assume that the underlying monotone function $g$ has a finite range $\mathcal{Y} \subseteq \mathbb{R}$.

\subsection{Pure DP -- Theorem \ref{thm:shi-pdp}}
    The simplest implementation of the shifted inverse mechanism is to apply the exponential mechanism \cite{mcsherry2007mechanism}. 
    Roughly, we want to find $y \in \mathcal{Y}$ such that $0<\ell(x,y)<2\tau$, where $\tau>0$ is an appropriately-chosen offset. The exponential mechanism can do this by minimising $|\ell(x,y)-\tau|$, which still has sensitivity $1$.

    However, $\ell$ is not a continuous function, so there may not exist \emph{any} $y$ such that $0<\ell(x,y)<2\tau$. Thus we need to be a bit more careful:
    For $\tau \in \mathbb{N}$, define \[\widehat{\ell}_\tau(x,y) := \max\{ \ell(x,y)-\tau, \tau-\overline{\ell}(x,y) \},\label{eq:quasiconcave}\] where $\ell(x,y) := \min \{ |x \setminus \widetilde{x}| : \widetilde{x} \subseteq x, g(\widetilde{x}) \le y\}$ and $\overline{\ell}(x,y) := \min \{ |x \setminus \widetilde{x}| : \widetilde{x} \subseteq x, g(\widetilde{x}) < y\} $ are as in Equations \ref{eq:shi} and \ref{eq:shi-loss2}.
    Intuitively, $\widehat{\ell}_\tau(x,y) \approx |\ell(x,y)-\tau|$, but we are guaranteed that there exists \emph{some} $y_*$ such that $\widehat{\ell}_\tau(x,y_*) \le 0$, namely for $y_* = \min\{ g(x') : x' \subseteq x, |x \setminus x'| = \tau \}$. Since $\widehat{\ell}_\tau$ has sensitivity $1$ in its first parameter, we can apply the exponential mechanism: That is, \[\forall x \in \mathcal{X}^* ~ \forall y \in \mathcal{Y} ~~~~~ \pr{}{M(x)=y} = \frac{\exp\big(-\frac\varepsilon2 \widehat{\ell}_\tau(x,y)\big)}{\sum_{\hat{y}\in\mathcal{Y}}\exp\big(-\frac\varepsilon2 \widehat{\ell}_\tau(x,\hat{y})\big)}\] defines a $(\varepsilon,0)$-differentially private algorithm $M : \mathcal{X}^* \to \mathcal{Y}$. 
    The exponential mechanism guarantees that \cite[Theorem 3.11]{dwork2014algorithmic} \[\pr{M}{\widehat{\ell}_\tau(x,M(x)) < \min_{y \in \mathcal{Y}} \widehat{\ell}_\tau(x,y) + \frac2\varepsilon\log(|\mathcal{Y}|/\beta)} \ge 1-\beta.\]
    Since $\min_{y \in \mathcal{Y}} \widehat{\ell}_\tau(x,y) \le 0$, setting $\tau=\lceil\frac2\varepsilon\log(|\mathcal{Y}|/\beta)\rceil$ yields the conclusion that, with probability at least $1-\beta$, we have $\tau-\overline{\ell}(x,M(x)) < \frac2\varepsilon\log(|\mathcal{Y}|/\beta)$ and $\ell(x,M(x)) - \tau < \frac2\varepsilon\log(|\mathcal{Y}|/\beta)$. 
    Now $\tau-\overline{\ell}(x,M(x)) < \frac2\varepsilon\log(|\mathcal{Y}|/\beta)$ implies $\overline{\ell}(x,M(x))>0$, which implies $M(x) \le g(x)$.
    Next $\ell(x,M(x)) - \tau < \frac2\varepsilon\log(|\mathcal{Y}|/\beta)$ implies $\ell(x,M(x)) \le 2\tau$, which implies $M(x) \ge \min\{ g(x') : x' \subseteq x, |x'| \ge |x|-2\tau\}$, as required.
    
\subsection{Approximate DP -- Theorem \ref{thm:shi-adp}}
    The main limitation of the exponential mechanism is its logarithmic dependence on the size of the search space $\mathcal{Y}$ -- i.e., $t = O(\log|\mathcal{Y}|)$.
    We can improve this by relaxing to approximate differential privacy (i.e, $(\varepsilon,\delta)$-differential privacy with $\delta>0$) and by exploiting the fact that $\ell(x,y)$ is a decreasing function of $y$. (The exponential mechanism does not exploit this structure.)
    At this point we can apply sophisticated algorithms from the literature. 

    To summarize, we have $\ell(x,y)$, which is sensitivity-1 in the private input $x$ and decreasing in the other input $y$, and our goal is to privately find $y_i \in \mathcal{Y}$ (where $\mathcal{Y} = \{y_1 < y_2 < \cdots < y_{|\mathcal{Y}|}\} \subseteq \mathbb{R}$) such that $\ell(x,y_i)\le t$ and $\ell(x,y_{i-1})>0$, where $y_{i-1}$ is the element in $\mathcal{Y}$ that is immediately before $y_i$ in sorted order.
    This formulation is exactly the \emph{generalized interior point problem} \citep{bun2018composable}. \citet{bun2018composable} sketch\footnote{Unfortunately, they do not provide a theorem statement for approximate differential privacy.} an algorithm for this task with $t = \frac1\varepsilon \log(1/\delta) \exp(O(\log^*|\mathcal{Y}|))$, as required. 

    Alternatively, we can formulate the problem as minimising a quasi-convex function: The function $\widehat{\ell}_\tau(x,y)$ from Equation \ref{eq:quasiconcave} has sensitivity $1$ in the private input $x$ and, in terms of the other input $y$, it is quasi-convex (i.e., decreasing-then-increasing). And our goal is to find an approximate minimiser $y \in \mathcal{Y}$. This is the formulation\footnote{There are some additional differences between their formulation and our formulation. E.g., they maximise a quasi-concave function; we minimise a quasi-convex function. Their statement \cite[Theorem 4.2]{cohen2023optimal} only guarantees success with probability $\ge 9/10$; this can be increased to $1$ either by modifying their algorithm or using generic reductions \cite{liu2019private,papernothyperparameter,DPorg-fail-prob}.} of \citet{cohen2023optimal} and they provide an algorithm to find an approximate minimiser with excess loss $t = \frac1\varepsilon \log(1/\delta) \exp(O(\log^*|\mathcal{Y}|))$, as required.
    
\subsection{Concentrated DP -- Theorem \ref{thm:shi-cdp}}
    The aforementioned sophisticated algorithms achieve a very good dependence on the size of the output space $\mathcal{Y}$, namely $t=\exp(O(\log^*|\mathcal{Y}|))$. The iterated logarithm is a function that is constant for all practical purposes; although it is, in theory, unbounded.
    However, these sophisticated algorithms are quite complicated and, to the best of our knowledge, have never been implemented. Furthermore, approximate differential privacy (i.e., $\delta>0$) can be undesirable, since it permits arbitrary privacy failures with probability $\le \delta$.
    Thus we also consider concentrated differential privacy \cite{dwork2016concentrated,bun2016concentrated}, which is a relaxation of pure differential privacy that does not permit arbitrary privacy failures; it also tends to correspond to more practical algorithms.

    A natural algorithm for performing the inversion \eqref{eq:shi-inversion} is binary search. To make this differentially private, we must use some form of \emph{noisy} binary search, which we discuss next.
    Noisy binary search gives us a $t=O(\sqrt{\log|\mathcal{Y}|})$ dependence on the size of the output space $\mathcal{Y}$.
    To the best of our knowledge, the specific guarantee in Theorem \ref{thm:shi-cdp} does not appear in the literature; hence we give more detail.

\subsection{Noisy Binary Search} \label{app:noisy-binary-search}
Noisy binary search appears frequently in the differential privacy literature \cite{blum2013learning,bun2017make,feldman2017generalization,smith2021non,aamand2025lightweight}.

\paragraph{Problem Statement:}
The setting is that we have a function $\ell : \mathcal{X}^* \times \mathcal{Y} \to \mathbb{R}$ that is low sensitivity in its first argument and decreasing in its second argument. That is, $|\ell(x,y)-\ell(x',y)| \le |x \setminus x'| + |x' \setminus x|$ and $y \le y' \implies \ell(x,y) \ge \ell(x,y')$. Informally, the goal is, given a private input $x$, to find $y$ such that $\ell(x,y) \approx \tau$ for some target value $\tau$.
To make this tractable we must restrict the search space $\mathcal{Y} \subseteq \mathbb{R}$ to be finite; namely, let $y_1 \le y_2 \le \cdots \le y_{|\mathcal{Y}|}$ be a sorted enumeration of the search space $\mathcal{Y}$.
To be precise, our goal is to find an index $i$ such that $\ell(x,y_{i-1}) > \tau - \eta$ and $\ell(x,y_i) < \tau + \eta$, where $i$ ranges from $i=1$ (in which case we define $\ell(x,y_{i-1})=\infty$) to $i=|\mathcal{Y}|+1$ (in which case we define $\ell(x,y_i)=-\infty$). Here $\eta>0$ is some tolerance. The output $i$ must be differentially private in terms of the input $x$.

\paragraph{Na\"ive Solution:}
We can find an appropriate index $i$ via binary search using $\log_2|\mathcal{Y}|$ steps, each time comparing the function value $\ell(x,y_i)$ to the target value $\tau$.
To ensure concentrated differential privacy, we add Gaussian noise $\mathcal{N}(0,\sigma^2)$ to the function value $\ell(x,y_i)$. This means that each comparison may be incorrect.
The parameter $\eta>0$ allows us to tolerate some amount of error.
By composition over $\log_2|\mathcal{Y}|$ steps, the scale of the noise grows as $\sigma = O(\sqrt{\log|\mathcal{Y}|})$.
Na\"ively, we must take a union bound over all $O(\log|\mathcal{Y}|)$ steps, which adds a $\sqrt{\log\log|\mathcal{Y}|}$ factor to our error guarantee -- i.e., $\eta = O(\sigma \cdot \sqrt{\log \log |\mathcal{Y}|}) = O(\sqrt{\log|\mathcal{Y}| \cdot \log \log |\mathcal{Y}|})$. (This is because the maximum of $k$ Gaussians is $\sqrt{\log k}$ standard deviations above the mean.) Fortunately, we can remove this asymptotic factor.

There are two ways to reduce this union bound factor: 
We could add non-independent noise that is carefully tailored to reduce the probability of one value being large \cite{Steinke_Ullman_2017,ganesh2021privately,dagan2022bounded,ghazi2021avoiding}. 
Alternatively, we can modify the binary search procedure itself to be noise-tolerant.

\paragraph{Binary Search Over Biased Coins:}
There is a rich literature on noise-tolerant versions of binary search, although most of this work considers settings that are not directly relevant to our setting.

\citet{karp2007noisy} consider a setting in which there are $m$ biased coins. The coins are sorted by bias, but otherwise the biases are unknown other than what we can learn by flipping the coins. The goal is to find a nearly unbiased coin by flipping the coins as few times as possible.\footnote{Each coin can be flipped multiple times and the outcomes are all independent. If no unbiased coin exists, the goal is to find a successive pair of coins whose biases are approximately on opposite sides of the unbiased threshold. In general, we can search of an arbitrary bias (instead of an unbiased coin).}
Our setting can be reduced to their setting.

\citet{karp2007noisy} give an algorithm that flips $O(\log m)$ coins and has a constant probability of success. (Obviously, the probability of success can be boosted by repetition.)
\citet{gretta2024sharp} give an improved algorithm with better constants and high-probability success bounds. 

\begin{theorem}[{\cite[Theorem 1.1]{gretta2024sharp}}]\label{thm:binarycoins}
    Let $\beta,\gamma\in(0,1/4)$ and $m \in \mathbb{N}$.
    Then there exists \[q = \frac{\log_2(m) + O\big((\log(m))^{2/3} (\log(1/\beta))^{1/3} + \log(1/\beta)\big)}{1 - H(1/2-\gamma)} \label{eq:noisy-bs}\] and an algorithm with the following properties. Here $H(p):=p\log_2(1/p)+(1-p)\log_2(1/(1-p))$ is the binary entropy function.
    
    Let $0 = p_0 \le p_1 \le p_2 \le \cdots \le p_m \le p_{m+1} = 1$.
    The algorithm has access to an oracle that, given an index $i \in [m]$, returns an independent sample from $\mathsf{Bernoulli}(p_i)$; otherwise the algorithm does not have access to $p_1, \cdots, p_m$.
    The algorithm makes $q$ queries to this oracle and, with probability at least $1-\beta$, returns $i \in [m+1]$ with $[p_{i-1},p_i] \cap (1/2-\gamma,1/2+\gamma) \ne \emptyset$.
\end{theorem}

We can simplify the bound \eqref{eq:noisy-bs} to $q\le O(\log(m/\beta)/\gamma^2)$.
\citet{gretta2024sharp} state a more general form of the result in which the target bias of $1/2$ can be set differently. The above special case of their result suffices for our application.

In our setting, the oracle returns a noisy value $V = \ell(x,y_i)+\mathcal{N}(0,\sigma^2)$. We can then compare this noisy value to the target threshold $\tau$. The binary indicator variable $\mathbb{I}[V>\tau]$ is then a Bernoulli random variable (i.e., a coin flip) with expectation \[\pr{}{\ell(x,y_i)+\mathcal{N}(0,\sigma^2) > \tau} = \pr{}{\mathcal{N}(0,1) > \frac1\sigma\big(\tau - \ell(x,y_i)\big)} = \frac12 + \Theta\left(\frac{\ell(x,y_i)-\tau}{\sigma}\right),\label{eq:threshold}\] where the asymptotic expression holds for $\ell(x,y_i) \approx \tau$.
Thresholding a noisy real value to a binary indicator loses information, but this suffices for our application. An avenue for future work is to improve noisy binary search to fully exploit the information given by noisy values. 

\begin{proposition}[DP Binary Search]\label{prop:binarysearch}
    Let $\ell : \mathcal{X}^* \times [m] \to \mathbb{R}$ have sensitivity $1$ in its first argument and be a decreasing function of its second argument -- i.e., $|\ell(x,y)-\ell(x',y)| \le | x \setminus x' | + | x' \setminus x |$ and $y \le y' \implies \ell(x,y) \ge \ell(x,y')$.
    Let $\rho, \beta > 0$ and $\tau \in \mathbb{R}$.
    Then there exists $\eta = O(\sqrt{\log(m/\beta)/\rho})$ and a $\rho$-zCDP \cite{bun2016concentrated} and $\sqrt{2\rho}$-GDP \cite{dong2022gaussian} algorithm $M : \mathcal{X}^* \to [m+1]$ with the following property.
    For all $x \in \mathcal{X}^*$, with probability at least $1-\beta$, we have $\ell(x,M(x)) < \tau + \eta$ and $\ell(x,M(x)-1) > \tau - \eta$, where we define $\ell(x,0)=\infty$ and $\ell(x,m+1)=-\infty$. 
\end{proposition}

Proposition \ref{prop:binarysearch} follows from Theorem \ref{thm:binarycoins} and basic properties of differential privacy (namely, composition and the Gaussian mechanism) \cite{steinke2022composition}. The algorithm asks $q=O(\log(m/\beta)/\gamma^2)$ queries in the form of an index $i \in [m]$ and gets answers in the form of samples $\mathcal{N}(\ell(x,y_i),\sigma^2)$.
To ensure $\rho$-zCDP and $\sqrt{2\rho}$-GDP we set $\sigma = \sqrt{q/2\rho}$.
Per Equation \ref{eq:threshold}, we can set $\gamma = \Theta\left(\frac{\eta}{\sigma}\right)$ to translate between the magnitude $\sigma$ of the noise added, the tolerance $\gamma$ in the coin bias, and the tolerance $\eta$ in the values. We set the tolerance in the coin biases $\gamma$ to some constant (e.g., $\gamma=1/5$).
Thus $\eta = \Theta(\sigma) = \Theta(\sqrt{q/\rho}) = \Theta(\sqrt{\log(m/\beta)/\rho})$.%, as required.

Theorem \ref{thm:shi-cdp} follows by combining Propositions \ref{prop:shi-loss} and \ref{prop:binarysearch}.
Proposition \ref{prop:shi-loss} gives a low-sensitivity/decreasing loss function. Proposition \ref{prop:binarysearch} gives a $\rho$-zCDP algorithm for approximately inverting it. Setting $\tau=\eta$ in Proposition \ref{prop:binarysearch} means, with probability at least $1-\beta$, we get an index $i$ such that $\ell(x,y_{i-1})>0$ and $\ell(x,y_i)<t=\tau+\eta=O(\sqrt{\log(|\mathcal{Y}|/\beta)/\rho})$.
The guarantee $\ell(x,y_{i-1})>0$ entails $y_i \le g(x)$, while $\ell(x,y_i)<t$ entails $y_i \ge \min\{ g(x') : x' \subseteq x, |x'| \ge |x|-t \}$, as required.

\section{Covering Designs} \label{app:coveringdesigns}

For completeness, we provide proofs of the bounds we use on the size of covering designs.

\begin{proposition}[Proposition \ref{prop:covering-size}]
    For all $n,m,t \in \mathbb{N}$ with $n \ge m \ge t$ we have
    \[\frac{{n \choose t}}{{m \choose t}} \le C(n,m,t) < \frac{{n \choose t}}{{m \choose t}} \left(1 + \log{ m \choose t} \right) + 1.\]
\end{proposition}
\begin{proof}
    The lower bound is due to \citet{Schonheim64} and the upper bound is due to \textcite[Theorem 13.4]{ErdosS74}.
    Both proofs rely on the probabilistic method.

    First the lower bound: 
    We claim that $C(n,m,t) \ge \frac{n}{m} C(n-1,m-1,t-1)$ for all $n \ge m \ge t \ge 1$.\footnote{Here we define $C(n,m,0)=1$ to make the induction work for $t=1$. (A $(n,m,0)$-covering design  needs one set $S_1 \subseteq [n]$ that ``contains'' the empty set $\emptyset \subseteq S_1$.)}
    Induction then gives \[C(n,m,t) \ge \frac{n}{m} C(n-1,m-1,t-1) \ge \frac{n}{m} \frac{n-1}{m-1} C(n-2,m-2,t-2) \ge \cdots \ge \prod_{i=0}^{t-1} \frac{n-i}{m-i} = \frac{{n \choose t}}{{m \choose t}} .\label{eq:schonheim}\]
    The claim follows from the following argument. 
    Let $S_1, S_2, \cdots, S_k \subseteq [n]$ be a $(n,m,t)$-covering design of size $k=C(n,m,t)$. Pick $U \in [n]$ uniformly at random.
    For each $i\in[k]$, if $U \in S_i$, we remove the element $U$ and call the new set $\widehat{S}_i = S_i \setminus \{U\}$; if $U \notin S_i$, we discard $S_i$.
    Next we renumber $[n]\setminus\{U\}$ to map to $[n-1]$ and reindex so that $\widehat{S}_1, \cdots, \widehat{S}_{\widehat k}$ excludes the discarded indices $i$ with $U \notin S_i$.
    Now $\widehat{S}_1, \cdots, \widehat{S}_{\widehat k} \subseteq [n-1]$ is a $(n-1,m-1,t-1)$-covering design. In particular, any $\widehat{T} \subseteq [n-1]$ of size $|\widehat{T}|\le t-1$ can be extended to $T \subseteq [n]$ with $|T|=|\widehat{T}|+1$ (by adding $U \in T$) such that there exists $i \in [k]$ with $T \subseteq S_i$, which implies $\widehat{T} \subseteq \widehat{S}_{\widehat{i}}$.
    The size $\widehat{k}$ of the new covering design depends on $U$. Specifically, $\widehat{k} = |\{ i \in [k] : U \in S_i \}|$. Since $U$ is uniformly random, we have $\ex{}{\widehat{k}} = \sum_{i \in [k]} \pr{}{U \in S_i} = \sum_{i \in [k]} \frac{|S_i|}{n} = \frac{m}{n} k$.
    There must exist a fixed choice of $U$ such that $\widehat{k} \le \frac{m}{n} k$, which rearranges to $C(n,m,t)=k \ge \frac{n}{m} \widehat{k} \ge \frac{n}{m} C(n-1,m-1,t-1)$, as required.

    Second the upper bound:
    Let $S_1, S_2, \cdots, S_{k_1} \subseteq [n]$ be independent and uniformly random of size $|S_i|=m$ for each $i \in [k_1]$.
    For any fixed $T \subseteq [n]$ of size $|T|=t$, we have \[\pr{}{\not\exists i \in [k_1] ~~ T \subseteq S_i} = \prod_{i \in [k_1]} \big( 1- \pr{}{T \subseteq S_i} \big) = \left( 1 - \frac{{n-t \choose m-t}}{{n \choose m}} \right)^{k_1} .\]
    Let $K_2$ be the number of sets $T \subseteq [n]$ of size $|T|=t$ such that there is no $i \in [k_1]$ with $T \subseteq S_i$.
    We have \[\ex{}{K_2} = {n \choose t} \pr{}{\not\exists i \in [k_1] ~~ T \subseteq S_i} = {n \choose t} \left( 1 - \frac{{n-t \choose m-t}}{{n \choose m}} \right)^{k_1},\] where $T \subseteq [n]$ with $|T|=t$ is arbitrary.\footnote{If $K_2=0$, then $S_1, \cdots, S_{k_1}$ form a covering design (i.e., there is no need to add more sets). Since $K_2 \ge 0$ is an integer, if $\ex{}{K_2}<1$, then $K_2=0$ happens with nonzero probability; this already proves $C(n,m,t) \le \frac{{n \choose t}}{{m \choose t}} \log{ n \choose t}$ \cite{chang1996set}. Going beyond existential results, if $\ex{}{K_2}$ is small (e.g., $\le 0.01$), then $S_1, \cdots, S_{k_1}$ form a covering design with high probability (e.g., $\ge 0.99$). \label{fn:cover-prob}}
    Now we can create a $(n,m,t)$-covering design of size $k_1+K_2$ by taking $S_1, \cdots, S_{k_1}$ and, for each uncovered $T \subseteq [n]$ of size $|T|=t$, adding an additional set that covers it.
    There must exist a fixed choice of the randomness such that $K_2 \le \lfloor \ex{}{K_2} \rfloor$. Thus we have
    \begin{align}
        C(n,m,t) &\le k_1 + \left\lfloor {n \choose t} \left( 1 - \frac{{n-t \choose m-t}}{{n \choose m}} \right)^{k_1} \right\rfloor \\
        &= k_1 + \left\lfloor {n \choose t} \left( 1 - \frac{{m \choose t}}{{n \choose t}} \right)^{k_1} \right\rfloor \tag{simplifying}\\
        &\le k_1 + \left\lfloor {n \choose t} \exp\left(  - k_1 \frac{{m \choose t}}{{n \choose t}} \right) \right\rfloor \tag{$1-x \le \exp(-x)$}. 
    \end{align}
    Setting $k_1 = \left\lceil \frac{{n \choose t}}{{m \choose t}} \log { m \choose t } \right\rceil$ yields \[C(n,m,t) \le \left\lceil \frac{{n \choose t}}{{m \choose t}} \log { m \choose t } \right\rceil + \left\lfloor \frac{{n \choose t}}{{m \choose t}} \right\rfloor < \frac{{n \choose t}}{{m \choose t}} \left( \log {m \choose t} + 1 \right) + 1,\] which is the desired result.\footnote{\citet{ErdosS74} state the result without the $+1$ at the end, but it is not clear to us how they obtain this result. Their proof ignores the need to round $k_1$ to an integer and they leave fixing this issue as an exercise.}
\end{proof}

We remark that the lower bound in Proposition \ref{prop:covering-size} can be improved:  Noting that $C(n,m,t)$ must be an integer, the inductive step can be improved to\\$C(n,m,t) \ge \left\lceil \frac{n}{m} C(n-1,m-1,t-1) \right\rceil$ for all $n \ge m \ge t \ge 1$. Equation \ref{eq:schonheim} then becomes \[C(n,m,t) \ge \left\lceil \frac{n}{m} C(n-1,m-1,t-1) \right\rceil \ge \left\lceil \frac{n}{m} \left\lceil \frac{n-1}{m-1} \left\lceil \frac{n-2}{m-2} \left\lceil \cdots \left\lceil \frac{n-t+1}{m-t+1} \right\rceil  \cdots \right\rceil \right\rceil \right\rceil \right\rceil\]

\begin{corollary}[Corollary \ref{cor:covering-size}]
    For all $n,m,t \in \mathbb{N}$ with $n \ge m \ge t>1$ we have
    \[
        \left(\frac{n}{m}\right)^t \le C(n,m,t) <  \left(\frac{n e}{m}\right)^t \cdot\min\{ 1 + m \log 2, 1 + t \log m \} +1.
    \]
    \iffalse
    \begin{align}
        \left(\frac{n}{m}\right)^t \le \left(\frac{2n-t+1}{2m-t+1}\right)^t \le C(n,m,t) &< \left(\frac{n}{m}\cdot\frac{n-t+1}{m-t+1}\right)^{t/2} \cdot\min\left\{\begin{array}{c}1+t+t\log(m/t),\\1+m\log2,\\1+t\log m\end{array}\right\}+1 \notag\\&\le \left(\frac{n-t+1}{m-t+1}\right)^t \cdot\min\{ 1 + m, 1 + t \log m \} +1.
    \end{align}
    \fi
\end{corollary}
\begin{proof}
    Proposition \ref{prop:covering-size} states that \[\frac{{n \choose t}}{{m \choose t}} \le C(n,m,t) < \frac{{n \choose t}}{{m \choose t}} \left(1 + \log{ m \choose t} \right)+1.\]
    We simplify these bounds, starting with the upper bound.
    We use standard bounds on binomial coefficients: For all $m \ge t > 1$, we have \[\left(\frac{m}{t}\right)^t \le {m \choose t} \le \min \left\{ \left(\frac{em}{t}\right)^t, m^t, 2^m \right\}.\]
    Thus $\log{ m \choose t} \le \min\{ t + t \log(m/t), t \log m, m \log 2\}$ and \[\frac{{n \choose t}}{{m \choose t}} \le \frac{\left(\frac{en}{t}\right)^t}{\left(\frac{m}{t}\right)^t} = \left(\frac{en}{m}\right)^t,\] which gives the upper bound.
    For the lower bound, write \[\frac{{n \choose t}}{{m \choose t}} = \prod_{i=0}^{t-1} \frac{n-i}{m-i} = \exp\left(\sum_{i=0}^{t-1} f(i) \right),\]
    where we define $f(x) = \log\left(\frac{n-x}{m-x}\right) = \log(n-x)-\log(m-x)$.
    For $0 \le x < m \le n$, we have $f'(x) = \frac{-1}{n-x}-\frac{-1}{m-x} = \frac{n-m}{(m-x)(n-x)} \ge 0$. 
    Thus $f$ is an increasing function on $[0,m)$. Using $\log(\frac{n}{m}) = f(0) \le f(i) \le f(t-1) = \log(\frac{n-t+1}{m-t+1})$ gives \[\left(\frac{n}{m}\right)^t \le \frac{{n \choose t}}{{m \choose t}} = \exp\left(\sum_{i=0}^{t-1} f(i) \right) \le  \left(\frac{n-t+1}{m-t+1}\right)^t,\]
    which yields the lower bound.
    For $0 \le x < m \le n$, we also have $f''(x) = \frac{-1}{(n-x)^2}-\frac{-1}{(m-x)^2} = \frac{(n-x)^2-(m-x)^2}{(n-x)^2(m-x)^2} = \frac{(n-m)((n-x)+(m-x))}{(n-x)^2(m-x)^2} \ge 0$. Thus $f$ is convex on $[0,m)$. This yields additional bounds:
    Jensen's inequality gives \[f\left(\frac{t-1}{2}\right) = f\left( \frac{1}{t}\sum_{i=0}^{t-1} i \right) \le \frac{1}{t}\sum_{i=0}^{t-1} f(i) \le \frac{1}{t}\sum_{i=0}^{t-1} \frac{(t-1-i)f(0)+if(t-1)}{t-1} = \frac{f(0) + f(t-1)}{2},\] which rearranges to \[\left(\frac{n-\frac{t-1}{2}}{m-\frac{t-1}{2}}\right)^t \le \frac{{n \choose t}}{{m \choose t}} = \exp\left(\sum_{i=0}^{t-1} f(i) \right) \le  \left(\frac{n}{m} \cdot \frac{n-t+1}{m-t+1}\right)^{t/2}.\] 
\end{proof}

Several bounds in Table \ref{tab:tradeoff} are derived from Corollary \ref{cor:covering-size}. In particular, if $n-m = \frac{cn}{t+c}$, then 
\begin{align}
    C(n,m,t) &< \left(\frac{ne}{m}\right)^t \cdot\min\{ 1 + m \log 2, 1 + t \log m \} +1  \\
    &= e^t \left(1 - \frac{n - m}{n}\right)^{-t}  \cdot\min\{ 1 + m \log 2, 1 + t \log m \} +1  \notag \\
    &= e^t \left(1 - \frac{c}{t+c}\right)^{-t}  \cdot\min\{ 1 + m \log 2, 1 + t \log m \} +1  \notag \\
    &= e^t \left(1 + \frac{c}{t}\right)^{t}   \cdot\min\{ 1 + m \log 2, 1 + t \log m \} +1  \notag \\
    &\le e^{t+c}   \cdot (1 + t \log m ) +1 = O(e^{t+c} \cdot t \log m) \tag{$1+x \le e^x$} .
\end{align}
Alternatively, Equation \ref{eq:chunks} gives \[C(n,m,t) = C(\ell(t+c),\ell t, t) \le C(t+c,t,t) = {t+c \choose t} \le (t+c)^c.\] Setting $\ell = \frac{n}{t+c} \in \mathbb{Z}$ yields $n-m = \ell c = \frac{cn}{t+c}$.

Corollary \ref{cor:covering-size} is not the tightest (asymptotic) bound possible. In particular, we use $\left(\frac{n}{m}\right)^t \le \frac{{n \choose t}}{{m \choose t}} \le \left(\frac{en}{m}\right)^t$ for all integers $1 \le t \le m \le n$. The gap between the upper and lower bounds is a factor of $e^t$, which can be significant. When $m \approx n$, the lower bound is tighter than the upper bound -- i.e., the $e^t$ factor becomes superfluous in this setting.
A tighter bound is \[\frac{{n \choose t}}{{m \choose t}} \le \left(c_{n,m} \frac{n}{m}\right)^t ~~~ \text{ where } ~~~ c_{n,m} := \left(1 + \frac{1}{n/m-1}\right)^{n/m-1} \in [1,e],\] which holds for all integers $1 \le t \le m < n$.

\end{document}